\documentclass[copyright,creativecommons]{eptcs}

\usepackage{iftex}

\ifpdf
  \usepackage{underscore}         
  \usepackage[T1]{fontenc}        
\else
  \usepackage{breakurl}           
\fi

\usepackage{defs}
\usepackage{mathrsfs}
\usepackage{amsfonts}
\usepackage{amssymb}
\usepackage{pgfplots}
\usepackage{multicol}
\usepackage{caption}
\usepackage{subcaption}

\newtheorem{example}{Example}

\newtheorem{proposition}{Proposition}

\newtheorem{lemma}{Lemma}

\newtheorem{definition}{Definition}

\title{On Imperfect Recall in Multi-Agent Influence Diagrams}
\author{James Fox
\institute{University of Oxford}
\email{james.fox@cs.ox.ac.uk}
\and
Matt MacDermott
\institute{Imperial College London}
\email{m.macdermott21@imperial.ac.uk}
\and
Lewis Hammond
\institute{University of Oxford}
\email{lewis.hammond@cs.ox.ac.uk}
\and
Paul Harrenstein
\institute{University of Oxford}
\email{paul.harrenstein@cs.ox.ac.uk}
\and
Alessandro Abate
\institute{University of Oxford}
\email{aabate@cs.ox.ac.uk}
\and
Michael Wooldridge
\institute{University of Oxford}
\email{mjw@cs.ox.ac.uk}
}
\def\titlerunning{On Imperfect Recall in MAIDs}
\def\authorrunning{J. Fox et al.}

\pgfplotsset{compat=1.17}
\begin{document}
\maketitle

\begin{abstract}
Multi-agent influence diagrams (MAIDs) are a popular game-theoretic model based on Bayesian networks. In some settings, MAIDs offer significant advantages over extensive-form game representations. Previous work on MAIDs has assumed that agents employ behavioural policies, which set independent conditional probability distributions over actions for each of their decisions. In settings with imperfect recall, however, a Nash equilibrium in behavioural policies may not exist. We overcome this by showing how to solve MAIDs with forgetful and absent-minded agents using mixed policies and two types of correlated equilibrium. We also analyse the computational complexity of key decision problems in MAIDs, and explore tractable cases. Finally, we describe applications of MAIDs to Markov games and team situations, where imperfect recall is often unavoidable.
\end{abstract}

\section{Introduction}

Multi-agent influence diagrams (MAIDs) are a graphical representation for dynamic non-cooperative games, which can be more compact and expressive than extensive-form games (EFGs) \citep{koller2003multi}. Like Bayesian networks (BNs), MAIDs use a directed acyclic graph (DAG) to represent conditional probabilistic dependencies between random variables, but they also specify decision and utility variables for each agent. Each agent selects a behavioural policy -- independent conditional probability distributions (CPDs) over actions for each of their decision variables -- to maximise their expected utility. A MAID's mechanised graph extends this DAG by explicitly representing each variable's distribution and showing which other variables' distributions matter to an agent optimising a particular decision rule \citep{causalgames,koller2003multi,dawid2002influence}. 

MAIDs, and their causal variants \citep{causalgames}, have been used in the design of safe and fair AI systems \citep{everitt2021agent,ashurst2022fair,everitt2021reward,farquhar2022path,carroll2023characterizing}, to explore reasoning patterns and deception \citep{pfeffer2007reasoning,ward2022agent}, and to identify agents from data \citep{kenton2022discovering}. However, to date, agents in MAIDs are usually assumed to have perfect (or, at least, `sufficient') recall~\citep{koller2003multi}. This assumption is often unreasonable. For example, MAIDs must allow imperfect recall to handle bounded rationality, teams with imperfect communication \citep{detwarasiti2005influence}, or memoryless policies in Markov games. However, forgetfulness (of previous observations) or absent-mindedness (about whether previous decisions have even been made) can prevent the existence of a Nash Equilibrium (NE) in behavioural policies. To overcome this, one can consider other solution concepts, such as mixed or correlated equilibria.

In this work, we focus on imperfect recall in MAIDs. Imperfect recall has already been extensively studied in EFGs \cite{piccione1997interpretation,Kuhn1953,waugh2009practical}, but a MAID's mechanised graph makes graphically explicit the semantic difference between behavioural and mixed policies (hidden in EFGs) and readily identifies forgetful or absent-minded agents (or teams).
Our insights inspire two definitions of \emph{correlated equilibrium} in MAIDs. The first follows from the normal-form game definition \cite{aumann1974subjectivity}. The second, based on von~Stengel and Forges' extensive-form correlated equilibrium \cite{von2008extensive}, is more natural for dynamic settings, can yield greater social welfare, and is easier to compute. Again, mechanised graphs clearly depict the assumptions made in both.
Next, we examine MAIDs from a computational complexity perspective by studying the decision problems of finding a best response, checking whether a policy profile is an NE, and checking whether each type of NE exists. These provide an insight into what makes particular instances hard, when computations can be made tractable, and rigorously identify which problems are suitable for analysis as MAIDs. Our results also apply to refinements of MAIDs, such as \textit{causal games} \citep{causalgames}.
We assume familiarity with EFGs \citep{maschler2020game}, BNs \citep{koller2009probabilistic}, and the complexity classes $\Poly$, $\NP$, and $\PP$~\citep{papadimitriou1994}. Proof sketches are provided, but details are deferred to the appendices. 

\paragraph{Related Work.} 
There is a rich literature on influence diagrams
\cite{kjaerulff2008bayesian} and imperfect recall has been studied in single-agent influence diagrams \cite{maua2012solving,maua2016fast,lauritzen2001representing,de2008strategy,van2022complete} as well as in EFGs \cite{aumann1997absent,kaneko1995behavior,Kuhn1953,piccione1997interpretation,waugh2009practical}. However, to our knowledge, we are the first to focus on imperfect recall in influence diagrams with multiple agents.

A full policy profile in a MAID induces a BN, so many of our results inherit from that setting, where the decision problem variant of marginal inference is, in general, \PP-complete \citep{littman2001stochastic}. However, we care about the cases we encounter in practice, not just the worst case. Marginal inference in a BN can be performed in time exponential in the treewidth of the underlying graph \citep{koller2009probabilistic}, which entails a poly-time algorithm when the treewidth is small. Similarly, we will see that tractable results for computations in MAIDs can be found when problems are restricted to certain settings. We also sometimes reduce from partial order games \cite{zahoransky2021partial}, which can be interpreted as MAIDs without chance nodes, with deterministic decision rules, and where each agent has a single utility node as a child of all the decision nodes.

\section{The Model}
\label{sec:background}

We use capital letters $V$ for random variables, lowercase letters $v$ for their instantiations, and bold letters $\bm{V}$ and $\bm{v}$, respectively, for sets of variables and their instantiations. 
We let $\dom(V)$ denote the (finite, non-singleton) domain of $V$ (for ease, we take this to be binary unless stated otherwise) and $\dom(\bm{V}) \coloneqq \bigtimes_{V \in \bm{V}}\dom(V)$. 
Parents and children of $V$ in a graph are denoted by $\Pa_V$ and $\Ch_V$, respectively (with $\pa_V$ and $\ch_V$ their instantiations) and
$\Delta(X)$ denotes the set of all probability distributions over a set $X$.

\begin{example}
  \label{ex:taxi}
An autonomous taxi decides whether to offer Alice a discount ($T$) depending on whether 
its journey count exceeds a quota ($Q$). Alice decides whether to accept a journey ($A$) depending on the price. The taxi wants to maximise profit, but if its journey count is less than the quota and Alice rejects it, the taxi pays a penalty (the municipality uses this mechanism to prevent a proliferation of unnecessary taxis). Alice's utility is a function of her decision and the price offered by the taxi.
\end{example}

Figure \ref{fig:taxi:a} shows a MAID for this example. Chance variables (moves by nature), decision variables, and utility variables are represented by white circles, squares, and diamonds, respectively. Full edges leading into chance and utility nodes represent probabilistic dependence, as in a BN. Dotted edges leading into decision nodes identify information available to the agent when a decision $D$ is made, so $\pa_D$, the values of $\Pa_D$, represents the decision context for $D$. 
In EFGs, imperfect information is represented using explicitly labelled information sets. In MAIDs, we can infer that Alice is unaware of the value of $Q$ when making her decision by the lack of edge $Q \rightarrow A$.  A parameterisation defines the CPDs for the chance and utility variables, whereas CPDs of decision nodes are chosen by the agents playing the game.

\begin{definition}[\citep{koller2003multi}]
    \label{def:MAID}
    A \textbf{multi-agent influence diagram (MAID)} is a structure $\model = (\graph, \bm{\theta})$. $\graph = (N, \bm{V}, E)$ specifies a set of agents $N = \{1,\dots,n\}$ and a DAG $(\bm{V}, E)$, where $\bm{V}$ is partitioned into chance variables $\bm{X}$, decision variables $\bm{D} = \bigcup_{i \in N} \bm{D}^i$, and utility variables $\bm{U} = \bigcup_{i \in N} \bm{U}^i$. 
    The parameters $\bm{\theta} = \{\theta_{V}\}_{V \in \bm{V} \setminus \bm{D}}$ define the CPDs $\Pr(V \mid \Pa_V)$ for each non-decision variable such that for any setting of the decision variables' CPDs, the resulting joint distribution over $\bm{V}$ is Markov compatible with the DAG, i.e., $\Pr(\bm{v}) = \prod_{V \in \bm{V}} \Pr(v \mid \pa_{V})$. 
\end{definition}

\begin{figure}[t]
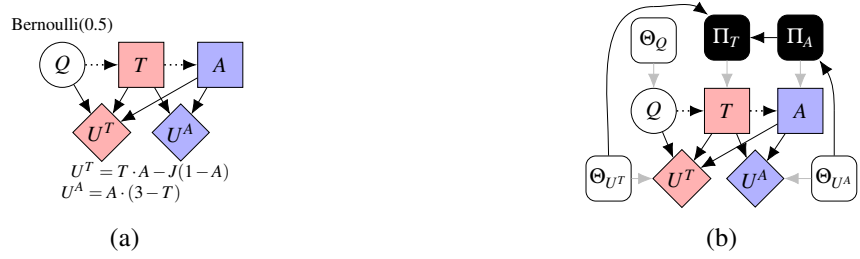

     \centering
     \begin{subfigure}[b]{0.49\linewidth}
               \vspace{-1em}
        \centering
        \resizebox{.75\width}{!}{
        \begin{influence-diagram}[every newellipse node/.style={inner sep=0pt}]
          \node (J) [] {$Q$};
          \node (D1) [decision, right = of J, player1] {$T$};
          \node (D2) [decision, right = of D1, player2] {$A$};
          \node (U1) [utility, below right = 1.2cm and 0.7cm of J, player1] {$U^T$};
          \node (U2) [utility, right = of U1, player2] {$U^A$};
          \edge {D1} {U1, U2};
          \edge {D2} {U1, U2};
          \edge [information] {J} {D1};
          \edge {J} {U1};
          \edge [information] {D1} {D2};
        
          \begin{scope}[
            every node/.style = {draw=none, rectangle, node distance=1mm}]
            \footnotesize
            \node [above = 0.05cm of J] {Bernoulli(0.5)};
            \node [below left = 0.15cm and -2.65cm of U1] {$U^T=T \cdot A-J(1-A)$};
            \node [below left = 0.5cm and -1.78cm of U1] {$U^A = A\cdot(3-T)$};

        \end{scope}
        
      \end{influence-diagram}
      }        
        \caption{}
      \label{fig:taxi:a}
     \end{subfigure}
     \begin{subfigure}[b]{0.49\linewidth}
               \vspace{-1em}
         \centering
          \resizebox{.75\width}{!}{
            \begin{influence-diagram}
        \node (J) [] {$Q$};
          \node (D1) [decision, right = 1.3cm of J, player1] {$T$};
          \node (D2) [decision, right = 1.3cm of D1, player2] {$A$};
          \node (Jmec) [relevancew, above = 1.25cm of J] {$\Theta_Q$};
          \node (D1mec) [relevanceb, above left = 1.3cm and 0cm of D1] {$\Pi_{T}$};
          \node (D2mec) [relevanceb, above left = 1.3cm and 0cm of D2] {$\Pi_{A}$};
          \node (U1) [utility, below right = 1.2cm and 0.5cm of J, player1] {$U^T$};
          \node (U2) [utility, below right = 1.2cm and 0.5cm of D1, player2] {$U^A$};

          \node (U1mec) [relevancew, left = 1.3cm of U1] {$\Theta_{U^T}$};
          \node (U2mec) [relevancew, right = 1.4cm of U2] {$\Theta_{U^A}$};

          \edge {D1} {U1, U2};
          \edge {D2} {U1, U2};
          \edge [information] {J} {D1};
          \edge {J} {U1};
          \edge [information] {D1} {D2};

          \edge [gray!50] {Jmec} {J};
          \edge [gray!50] {D1mec} {D1};
          \edge [gray!50] {D2mec} {D2};

          \edge [gray!50] {U1mec} {U1};
          \edge [gray!50] {U2mec} {U2};
          
          \edge {D2mec} {D1mec};

          \node (space1) [minimum size=0mm, node distance=2mm, above = 0.6cm of Jmec, draw=none] {};
      \draw (U1mec) edge[,in=-180,out=90] (space1.center)
        (space1.center) edge[,->,out=0,in=130] (D1mec);

        \node (space2) [minimum size=0mm, node distance=2mm, right = 0.6cm of D2, draw=none] {};
      \draw (U2mec) edge[,in=-90,out=90] (space2.center)
        (space2.center) edge[,->,out=90,in=-50] (D2mec);

        \end{influence-diagram}
      }   
         \caption{}
         \label{fig:taxi:b}
     \end{subfigure}

\caption{A MAID (a) and its mechanised graph (b) for Example \ref{ex:taxi}, which is a perfect recall and imperfect, but sufficient, information game.}
  \label{fig:taxi}
\end{figure}

Given a MAID, a \textbf{decision rule} $\pi_D$ for $D \in \bm{D}$ is a CPD $\pi_D(D \mid \Pa_D)$. A \textbf{partial (behavioural) policy profile} $\pi_{\bm{D}'}$ is a set of decision rules for each $D \in \bm{D}' \subseteq \bm{D}$, whereas $\pi_{-\bm{D}'}$ is the set of decision rules for each $D \in \bm{D} \setminus \bm{D}'$. 
A \textbf{(behavioural) policy} ${\bm{\pi}}^i$ refers to ${\bm{\pi}}_{\bm{D^i}}$, and a \textbf{(full) policy profile} ${\bm{\pi}} = ({\bm{\pi}}^1,\ldots,{\bm{\pi}}^n)$ is a tuple of policies, where ${\bm{\pi}}^{-i} \coloneqq ({\bm{\pi}}^1, \dots, {\bm{\pi}}^{i-1}, {\bm{\pi}}^{i+1}, \dots, {\bm{\pi}}^n)$. 
A decision rule is \textbf{pure} if $\pi_D(d \mid \pa_D) \in \{0,1\}$, which holds for a policy (profile) if it holds for all decision rules in the policy (profile).
For clarity, we use an overhead dot to mark this determinism, e.g., $\dot{\pi}_D, \dot{{\bm{\pi}}}^i$, or $\dot{{\bm{\pi}}}$.

By combining ${\bm{\pi}}$ with the partial distribution $\Pr$ over the chance and utility variables, we obtain a joint distribution: 
\[
\scalebox{.96}[1]{
$
\Pr^{\bm{\pi}}(\bm{x},\bm{d},\bm{u}) \coloneqq \prod_{V \in \bm{V} \setminus \bm{D}} \Pr (v \mid \pa_{V}) \cdot \prod_{D \in \bm{D}} \pi_{D} (d \mid \pa_{D})
$
}
\]

\noindent A full policy profile ${\bm{\pi}}$ therefore induces a BN with DAG given by the MAID's graph. Agent $i$'s \textbf{expected utility} $EU^i({\bm{\pi}})$ for a given policy profile ${\bm{\pi}}$ is defined as the expected sum of their utility variables: 
\[\textstyle
EU^i({\bm{\pi}}) \coloneqq \sum_{U \in \bm{U}^i} \sum_{u \in \dom(U)} \Pr^{{\bm{\pi}}} (U = u)\cdot u
\] 

Utility variables have deterministic CPDs, so can be interpreted as functions $U:\dom(\Pa_U) \rightarrow~\mathbb{R}$ to show their functional dependence on their parents (e.g., Figure \ref{fig:taxi:a}). An NE is defined in the usual way.

\begin{definition}[\citep{koller2003multi}] 
  \label{def:NE}
  A (behavioural) policy profile ${\bm{\pi}}$ is a \textbf{Nash equilibrium (NE)} (in behavioural policies) if for every agent $i \in N$ and every alternative (behavioural) policy ${\bm{\varpi}}^i$: $EU^i({\bm{\pi}^{-i}, {\bm{\pi}}^{i}}) \geq EU^i({\bm{\pi}^{-i}, {\bm{\varpi}}^{i}})$
\end{definition}

Collectively, the decision rules of decision variables and the CPDs of chance or utility nodes are known as mechanisms.
A mechanism $\mecvar_V$ for $V$ is \textbf{strategically relevant} to a decision rule for $D$ if the choice of the CPD at $\mecvar_V$ can affect the optimal choice of this decision rule. Koller and Milch \cite{koller2003multi} define an associated sound and complete graphical criterion for strategic relevance, $\bm{s}$-\textbf{reachability}, based on d-separation which can be checked in $\mathcal{O}(|\bm{V}| + |E|)$ time \citep{shachter2013bayes} (see Appendix \ref{app:stratrel} for formal definitions).

A MAID's regular graph $\graph$ captures the probabilistic dependencies between \textbf{object-level} variables in the game's environment, but its \textbf{mechanised graph} $\mec{\graph}$ is an enhanced representation which adds an explicit representation of the strategically relevant dependencies between agents' decision rules and the game's parameterisation (see  \citep{causalgames} for details). Each object-level variable $V \in \bm{V}$ has a mechanism parent $\mecvar_V$ representing the distribution governing $V$: each decision $D$ has a new \emph{decision rule} parent $\Pi_D = \mecvar_D$ and each non-decision $V$ has a new \emph{parameter} parent $\Theta_V = \mecvar_V$, whose values parameterise the CPDs.

Agents select a decision rule $\pi_D$ (i.e., the value of a decision rule variable $\Pi_D$) based on both the parameterisation of the game (i.e., the values of the parameter variables) and the selection of the other decision rules ${\bm{\pi}}_{-D}$ -- these dependencies are captured by the edges from other mechanisms into decision rule nodes. $s$-reachability determines which of these edges are necessary, so $\mecvar_V \rightarrow \Pi_D$ exists if and only if $\Pi_D$ strategically relies on $\mecvar_V$. The mechanised graph for Example \ref{ex:taxi} (in Figure \ref{fig:taxi:b}) shows that $\Pi_{T}$ strategically relies on $\Theta_{U^T}$ and $\Pi_{A}$, whereas $\Pi_{A}$ only strategically relies on $\Theta_{U^A}$. In contrast to a MAID's regular graph $\graph$, which is a DAG, there may exist cycles between mechanisms (e.g., Figure \ref{fig:noNE:a}).

For convenience, we denote the set of agent~$i$'s behavioural policies as ${\bm{P}}^i \coloneqq \dom(\bm{\Pi}^i)$, with sets of pure policies denoted as $\dot{{\bm{P}}}^i$ and (pure) policy profiles denoted by ${\bm{P}}$ ($\dot{{\bm{P}}}$).

\subsection{Concise Representations}
\label{sec:rep}

A concise representation of MAIDs is needed for three reasons. First, real numbers may obscure the true complexity of the problems \citep{bodlaender2002complexity}, so we assume that all probability parameters are given by a fraction of
two integers, both expressed in finite binary notation. This is realistic since the probabilities are normally either assessed by domain experts or estimated by a learning algorithm and means that all CPDs can be read in poly-time. Second, even with binary variables, a joint distribution across $\bm{V}$ requires $2^{|\bm{V}|}-1$  parameters. A MAID or BN's graphical Markov factorisation reduces this to $\sum_{V \in \bm{V}}2^{|\Pa_V|}$, but this can still be exponential in $|\bm{V}|$. Therefore, it is standard \citep{shimony1994finding,roth1996hardness,kwisthout2009computational,koller2009probabilistic} to assume that the maximum in-degree in the graph is much less than $|\bm{V}|$ (or constant), so that the size of the CPDs are polynomial in $|\bm{V}|$. This means that the total representation of our MAID (including all CPDs) is polynomial in our chosen complexity parameter $|\bm{V}|$.
Finally, as in BNs, our complexity results are strongly affected by the DAG's \textbf{treewidth}. The \textbf{treewidth} of a DAG measures its resemblance to a tree and is given by the number of vertices in the largest clique of the corresponding triangulated moral graph minus one \citep{bodlaender1993linear}.

\section{Imperfect Recall in MAIDs}
\label{sec:info}

Agents may possess different degrees of information about the state of a game. A game has \textbf{perfect recall} if each agent remembers all their past decisions and observations, and it has \textbf{perfect information} if each agent is aware of \textit{every} agent's past decisions and observations.

\begin{definition}[\citep{koller2003multi}]
  \label{def:perfect} Agent $i$ in a MAID $\model$ is said to have \textbf{perfect recall}  if there exists a total ordering $D_1 \prec \cdots \prec D_m$ over $\bm{D}^i$ such that $(\Pa_{D_j} \cup D_j) \subseteq \Pa_{D_k}$ for any $1 \leq j < k \leq m$. $\model$ is a perfect recall game if all agents in $\model$ have perfect recall. $\model$ is a \textbf{perfect information} game if there exists such an ordering over $\bm{D}$.
\end{definition}

A MAID with perfect information (recall) can be transformed into an EFG with perfect information (recall), and vice versa \cite{Hammond2021}. Hence, these information conditions also guarantee the existence of an NE in pure (behavioural) policies in the MAID (\cite{Kuhn1953} gives the equivalent results in EFGs).
However, the mechanised representation of a MAID enables weaker criteria to be defined -- \textbf{sufficient information} and \textbf{sufficient recall}. Later, in Proposition \ref{prop:suffexistence}, we will see that these criteria preserve  
the NE existence results of perfect information and perfect recall games, respectively. 

\begin{definition}
  \label{def:sufficient}
  Agent $i$ in a MAID $\model$ has \textbf{sufficient recall} \citep{milch2008ignorable} if the subgraph of the mechanised graph $\mec{\graph}$ restricted to just agent $i$'s decision rule nodes $\bm{\Pi}_{\bm{D}^i}$ is acyclic. $\model$ is a sufficient recall game if all agents in $\model$ have sufficient recall. $\model$ is a \textbf{sufficient information} game if the subgraph of $\mec{\graph}$ restricted to contain only and all decision rule nodes $\bm{\Pi}_{\bm{D}}$ is~acyclic.\footnote{Note that since previous work on influence diagrams has not modelled absent-mindedness (see our Definition \ref{def:forgetabsent} in Section~\ref{sec:forgetabsent}), this definition implicitly assumes each mechanism variable has a single child.}
\end{definition}

\subsection{Forgetfulness and Absent-Mindedness}
\label{sec:forgetabsent}

\begin{figure}[t]
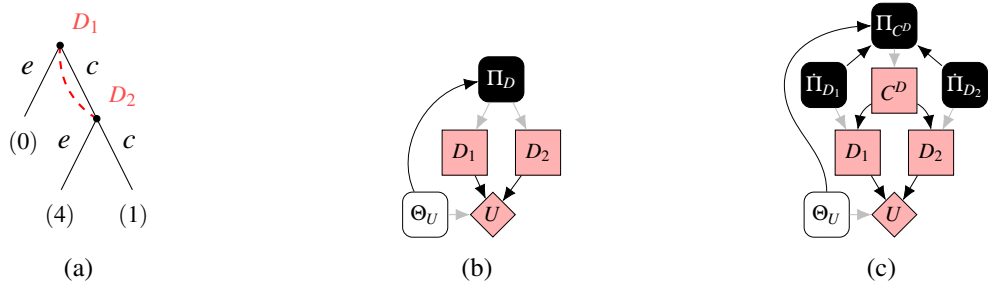

     \centering
\begin{subfigure}[b]{0.32\linewidth}
      \vspace{-1em}
      \centering
      \begin{istgame}[scale=0.65]
          \xtdistance{15mm}{15mm}
          \istroot(0)(0,0)<45, red!70>{$D_1$}
          \istb{e}[al]{(0)}
          \istb{c}[ar]{}
          \endist
          \istroot(1)(0-2)<45, red!70>{$D_2$}
          \istb{e}[al]{(4)}
          \istb{c}[ar]{(1)}
          \endist
          \setxtinfosetstyle{dashed,red,thick}
          \xtCInfoset(0)(1)<0.7>{\textcolor{red!70}{}}[below left]
      \end{istgame}
      \caption{}
      \label{fig:absent-minded:a}
  \end{subfigure}
  \begin{subfigure}[b]{0.33\linewidth}
  \vspace{-1em}
  \centering
  \resizebox{0.75\width}{!}{
    \begin{influence-diagram}[rotate=-90] 
      \node (D11) [decision, player1] {$D_1$};
      \node (D12) [decision, right = 1.3cm of D11, player1] {$D_2$};
      \node (Dmec) [relevanceb, above right = 1.3cm and 0.65cm of D11] {$\Pi_D$};
      \node (U1) [utility, below right = 1.1cm and 0.5cm of D11, player1] {$U$};
      \edge [gray!50] {Dmec} {D11, D12};
      \edge {D11, D12} {U1};
      \node (U1mec) [relevancew, left = 1.2cm of U1] {$\Theta_{U}$};
      \path (U1mec) edge[->, bend left=55] (Dmec);
      \edge [gray!50] {U1mec} {U1};
    \end{influence-diagram}
  }
  \caption{}
  \label{fig:absent-minded:b}
\end{subfigure}
  \begin{subfigure}[b]{0.33\linewidth}
        \vspace{-1em}
      \centering
      \resizebox{0.75\width}{!}{
      \begin{influence-diagram}
      \node (D11) [decision,player1 ] {$D_1$};
      \node (D12) [decision, player1, right = 1.3cm of D11] {$D_2$};
      \node (M) [decision, player1, above right = 1.1cm and 0.65cm of D11] {$C^D$};
      \node (U1) [utility, below right = 1.1cm and 0.65cm of D11, player1] {$U$};
      \node (Mmec) [relevanceb, above = 1.15cm of M] {$\Pi_{C^D}$};
      \node (D11mec) [relevanceb, above left = 1.2cm and 0.6cm of D11] {$\dot{\Pi}_{D_1}$};
      \node (D12mec) [relevanceb, above right = 1.2cm and 0.6cm of D12] {$\dot{\Pi}_{D_2}$};
      \edge {D11, D12} {U1};
      \edge [gray!50] {Mmec} {M};
      \node (U1mec) [relevancew, left = 1.2cm of U1] {$\Theta_{U}$};
      \edge [gray!50] {U1mec} {U1};
      \edge [gray!50] {D11mec} {D11};
      \edge [gray!50] {D12mec} {D12};
      \edge {D12mec, D11mec} {Mmec};

      \draw (M) edge[,->,in=90,out=-30] (D12);
      \draw (M) edge[,->,in=90,out=-150] (D11);

      \node (space1) [minimum size=0mm, node distance=2mm, left = 0.7cm of D11mec, draw=none] {};
      \draw (U1mec) edge[,in=-90,out=90] (space1.center)
        (space1.center) edge[,->,out=90,in=180] (Mmec);

      \end{influence-diagram}}
      \caption{}
      \label{fig:absent-minded:c}
  \end{subfigure}
\caption{The EFG (a) and the mechanised graphs for an absent-minded driver choosing behavioural (b) or mixed (c) policies.}
  \label{fig:absentminded}
\end{figure}

Previous work on MAIDs has assumed perfect or sufficient recall. We now begin the contributions of this paper by distinguishing between two types of imperfect recall in MAIDs.
\textbf{Forgetfulness} applies when an agent forgets an observation or the \emph{outcome} of one of their previous decisions. \textbf{Absent-mindedness} applies when an agent cannot even remember whether they have previously made a decision. To make this distinction, we leverage the following insight: \emph{mechanism nodes represent the CPDs governing object-level variables. Every edge between a mechanism and object-level node represents an independent draw from the mechanism's distribution.} We now provide formal definitions.

\begin{definition}
  \label{def:forgetabsent}
    Agent $i$ has \textbf{imperfect recall} in a MAID $\model$ if for every total ordering $D_1 \prec \cdots \prec D_m$ over $\bm{D}^i$ there exists some $j < k$ such that $(\Pa_{D_j} \cup D_j) \not\subseteq \Pa_{D_k}$ (i.e., if agent $i$ does not have perfect recall). Agent $i$ is \textbf{forgetful} if such a $D_j$ and $D_k$ have distinct decision rules and is \textbf{absent-minded} if in $\model$'s mechanised graph, a decision rule node has more than one outgoing edge to a decision node. 
\end{definition}

To motivate our definition of absent-mindedness in MAIDs, we revisit Piccione and Rubinstein's absent-minded driver game  \cite{piccione1997interpretation} (its EFG is in Figure \ref{fig:absent-minded:a}). A driver on a highway may take one of two exits. Taking the first, second, or no exit yields a payoff of~0,~4, or~1, respectively. Adopting Aumann \cite{aumann1997absent}'s \emph{modified multi-selves approach} (i.e., that the driver should only be able to control her current action, not her future actions), the driver does not know which junction she is facing, so she must have the same decision rule at both junctions. We make absent-mindedness explicit with a shared decision rule node $\Pi_D$ for $D_1$ and $D_2$ in the mechanised graph (Figure \ref{fig:absent-minded:b}) (note this is consistent with our mechanised graph definition). $\Pi_D$'s \emph{two outgoing edges now represent two independent draws from the same distribution}. For $D_i$ and $D_j$ to share a decision rule, it is necessary that $dom(D_i)=dom(D_j)$ and $dom(\emph{\Pa}_{D_i})=dom(\emph{\Pa}_{D_j})$. Note that perfect recall implies that for any two decisions belonging to the same agent, one's set of parents is a strict superset of the other's, so their decision rules have a different type signature, which rules out absent-mindedness.

In the following examples, used just to explain this paper's concepts, Alice and Bob play variations of matching pennies with the usual payoffs given according to the \emph{final} state of their two coins (where $a/b$ and $\bar{a}/\bar{b}$ represent heads and tails, respectively). Example \ref{ex:forget} illustrates a consequence of Bob being forgetful -- meaning he cannot remember the \textit{outcome} of his previous decision. In Example \ref{ex:absentminded}, Bob is absent-minded -- he cannot remember whether he has made a decision at all.

\begin{example}[Figures~\ref{fig:noNE:a}-\ref{fig:noNE:c}]
  \label{ex:forget}
    Bob is told he must submit a move in advance ($B_1$) and then confirm it on game day ($B_2$). If his moves agree, payoffs correspond with normal matching pennies, but if his moves disagree, he must forfeit and always loses (these payoffs are shown in Figure \ref{fig:noNE:c}). Bob is forgetful, so on game day he cannot remember his advance choice (i.e., the edge $B_1 \rightarrow B_2$ is missing in Figure \ref{fig:noNE:a}).
\end{example}

\begin{example}[Figures~\ref{fig:noNE:d}-\ref{fig:noNE:f}]
  \label{ex:absentminded}
  In a new game, the pennies start heads up, and Bob decides whether or not to turn the coin over ($B_1$). He is absent-minded, so when he sees heads he cannot remember whether he has already made his move, and he decides again ($B_2$). If he turns the coin having previously chosen to keep heads, Bob gets a $-2$ penalty and Alice a $+2$ bonus. In all other cases, the payoffs correspond with normal matching pennies (payoffs are shown at the leaves of the EFG in Figure~\ref{fig:noNE:e}).
\end{example}

Observe that the MAID's regular graph (just the object-level variables) is identical for both Figures~\ref{fig:noNE:a} and \ref{fig:noNE:d} with the missing $B_1 \rightarrow B_2$ edge implying imperfect recall. The difference between forgetfulness and absent-mindedness is only revealed by the mechanised graph. Forgetful Bob has two independent decision rules $\Pi_{B_1}$ and $\Pi_{B_2}$ for $B_1$ and $B_2$. Absent-minded Bob only has one shared decision rule $\Pi_{B}$.

\begin{figure}[tb]
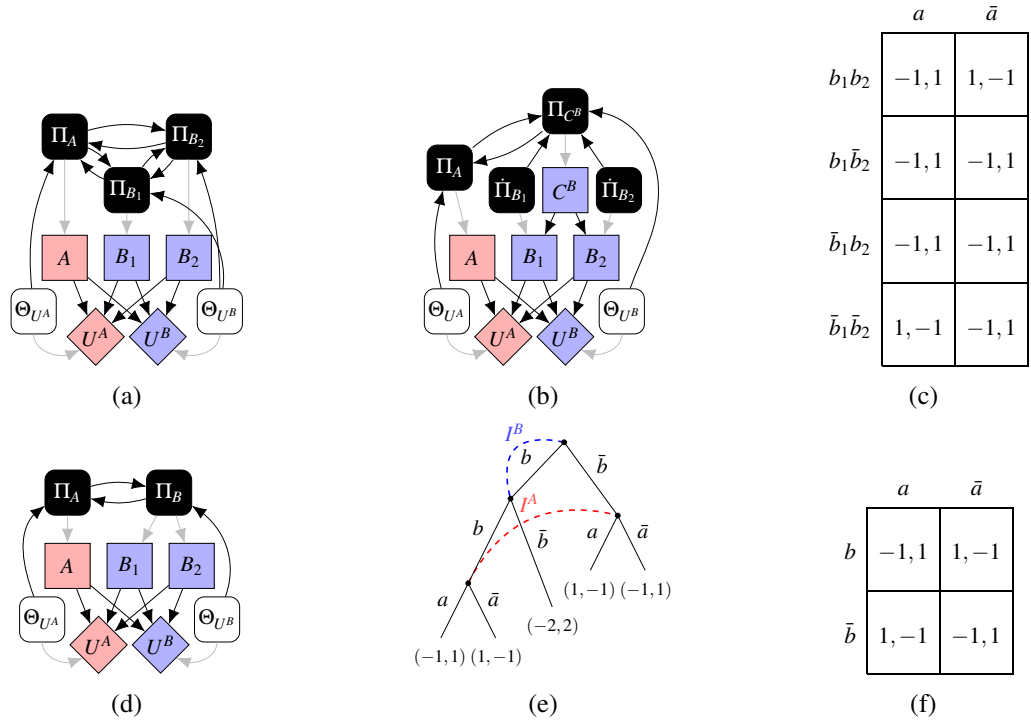

  \centering
  \begin{subfigure}[b]{0.34\linewidth}
        \vspace{-1em}
      \centering
      \resizebox{0.75\width}{!}{
        \begin{influence-diagram}
          \node (D1) [decision, player1] {$A$};
          \node (D21) [decision, right = 1.1cm of D1, player2] {$B_1$};
          \node (D22) [decision, right = 1.1cm of D21, player2] {$B_2$};
          \node (D1mec) [relevanceb, above right = 2.15cm and 0cm of D1] {$\Pi_A$};
          \node (D21mec) [relevanceb, above = 1.2cm of D21] {$\Pi_{B_1}$};
          \node (D22mec) [relevanceb, above left = 2.15cm and 0cm of D22] {$\Pi_{B_2}$};
          \node (U1) [utility, below right = 1.4cm and 0.55cm of D1, player1] {$U^A$};
          \node (U2) [utility, below right = 1.4cm and 0.55cm of D21, player2] {$U^B$};

          \node (U1mec) [relevancew, above left = 0.45cm and 1.1cm of U1] {$\Theta_{U^A}$};
          \node (U2mec) [relevancew, above right = 0.45cm and 1.1cm of U2] {$\Theta_{U^B}$};

          \edge {D1} {U1, U2};
          \edge {D21} {U1, U2};      
          \edge {D22} {U1, U2};
          \edge [gray!50] {D1mec} {D1};
          \edge [gray!50] {D21mec} {D21};
          \edge [gray!50] {D22mec} {D22};

          \path (D1mec) edge[->, bend left=15] (D21mec);
          \path (D21mec) edge[->, bend left=15] (D22mec);
          \path (D21mec) edge[->, bend left=15] (D1mec);
          \path (D1mec) edge[->, bend left=15] (D22mec);
          \path (D22mec) edge[->, bend left=15] (D21mec);
          \path (D22mec) edge[->, bend left=15] (D1mec);

          \path (U1mec) edge[->, bend left=15] (D1mec);
          \path (U2mec) edge[->, bend right=10] (D22mec);
          \path (U2mec) edge[->, bend right=45] (D21mec);

          \draw [gray!50] (U1mec) edge[,->,in=-140,out=-90] (U1);

          \draw [gray!50] (U2mec) edge[,->,in=-40,out=-90] (U2);

      \end{influence-diagram}
      }
      \caption{}
      \label{fig:noNE:a}
  \end{subfigure}
  \begin{subfigure}[b]{0.34\linewidth}
        \vspace{-1em}
    \centering
    \resizebox{0.75\width}{!}{
      \begin{influence-diagram}
        \node (D1) [decision, player1] {$A$};
        \node (D21) [decision, right = 1.1cm of D1, player2] {$B_1$};
        \node (D22) [decision, right = 1.1cm of D21, player2] {$B_2$};
        \node (M) [decision, above right = 1.2cm and 0.55cm of D21, player2] {$C^B$};
        \node (Mmec) [relevanceb, above = 1.4cm of M] {$\Pi_{C^B}$};
        \node (D1mec) [relevanceb, above left = 1.6cm and 0.4cm of D1] {$\Pi_A$};
        \node (U1) [utility, below right = 1.4cm and 0.55cm of D1, player1] {$U^A$};
        \node (U2) [utility, below right = 1.4cm and 0.55cm of D21, player2] {$U^B$};
        \edge {D1} {U1, U2};
        \edge {D21} {U1, U2};      
        \edge {D22} {U1, U2}; 
        \edge {M} {D21, D22}; 
        \edge [gray!50] {Mmec} {M};
        \edge [gray!50] {D1mec} {D1};
        \path (Mmec) edge[->, bend left=15] (D1mec);
        \path (D1mec) edge[->, bend left=15] (Mmec);

        \node (D21mec) [relevanceb, above left = 1.2cm and 0.4cm of D21] {$\dot{\Pi}_{B_1}$};
        \node (D22mec) [relevanceb, above right = 1.2cm and 0.4cm of D22] {$\dot{\Pi}_{B_2}$};
        \edge [gray!50] {D21mec} {D21};
        \edge [gray!50] {D22mec} {D22};
        \edge {D21mec, D22mec} {Mmec};

        \node (U1mec) [relevancew, above left = 0.45cm and 1.0cm of U1] {$\Theta_{U^A}$};
        \node (U2mec) [relevancew, above right = 0.45cm and 1.0cm of U2] {$\Theta_{U^B}$};
        \path (U1mec) edge[->, bend left=15] (D1mec);

        \node (space1) [minimum size=0mm, node distance=2mm, right = 0.7cm of D22mec, draw=none] {};
        \draw (U2mec) edge[,in=-90,out=70] (space1.center)
          (space1.center) edge[,->,out=90,in=0] (Mmec);

        \draw [gray!50] (U1mec) edge[,->,in=-140,out=-90] (U1);

        \draw [gray!50] (U2mec) edge[,->,in=-40,out=-90] (U2);

    \end{influence-diagram}
    }
    \caption{}
    \label{fig:noNE:b}
\end{subfigure}
\begin{subfigure}[b]{0.28\linewidth}
      \vspace{-1em}
  \centering
    \scalebox{.8}{$
    \begin{array}{r|c|c|c|c}
          \multicolumn{1}{c}{\rule[-1.2ex]{0pt}{3ex}}						&	\multicolumn{1}{c}{\mathwordbox{a}{xxxxx}}	&	\multicolumn{1}{c}{\mathwordbox{\bar{a}}{xxxxx}}	\\\cline{2-3}
          \rule[-3ex]{0pt}{8ex}b_1b_2			&	-1,1					& 1,-1												\\\cline{2-3}
          \rule[-3ex]{0pt}{8ex}b_1\bar{b}_2			&	-1,1					&	-1,1												\\\cline{2-3}
          \rule[-3ex]{0pt}{8ex}\bar{b}_1	{b}_2		&	-1,1					&	-1,1											\\\cline{2-3}
          \rule[-3ex]{0pt}{8ex}\bar{b}_1	\bar{b}_2		&	1,-1					&	-1,1											\\\cline{2-3}
      \end{array}
    $}
    \caption{}
    \label{fig:noNE:c}
  \end{subfigure}
  
\begin{subfigure}[b]{0.34\linewidth}
  \centering
  \resizebox{0.75\width}{!}{
    \begin{influence-diagram}
      \node (D1) [decision, player1] {$A$};
      \node (D21) [decision, right = 1.1cm of D1, player2] {$B_1$};
      \node (D22) [decision, right = 1.1cm of D21, player2] {$B_2$};
      \node (D1mec) [relevanceb, above = 1.3cm of D1] {$\Pi_A$};
      \node (D2mec) [relevanceb, above right = 1.3cm and 0.7cm of D21] {$\Pi_{B}$};
      \node (U1) [utility, below right = 1.4cm and 0.55cm of D1, player1] {$U^A$};
      \node (U2) [utility, below right = 1.4cm and 0.55cm of D21, player2] {$U^B$};
      \edge {D1} {U1, U2};
      \edge {D21} {U1, U2};      
      \edge {D22} {U1, U2};
      \edge [gray!50] {D1mec} {D1};
      \edge [gray!50] {D2mec} {D21, D22};

      \path (D1mec) edge[->, bend left=15] (D2mec);
      \path (D2mec) edge[->, bend left=15] (D1mec);

      \node (U1mec) [relevancew, above left = 0.45cm and 1.0cm of U1] {$\Theta_{U^A}$};
        \node (U2mec) [relevancew, above right = 0.45cm and 1.0cm of U2] {$\Theta_{U^B}$};
    \draw [gray!50] (U1mec) edge[,->,in=-140,out=-90] (U1);
    \draw [gray!50] (U2mec) edge[,->,in=-40,out=-90] (U2);

    \draw (U1mec) edge[,->,in=-150,out=110] (D1mec);
    \draw (U2mec) edge[,->,in=-30,out=70] (D2mec);

  \end{influence-diagram}
  }
  \caption{}
  \label{fig:noNE:d}
\end{subfigure}
\begin{subfigure}[b]{0.34\linewidth}
\centering
\resizebox{0.75\width}{!}{
\begin{istgame}
  \istroot(0){}+10mm..19mm+
  \istb{b}[al]{} \istbA(1.3){\bar{b}}[ar]{} \endist
  \istroot(1)(0-1) \istb{b}[al]{} \istbA(1.3){\bar{b}}[ar]{(-2,2)} \endist
  \xtdistance{10mm}{10mm}
  \istroot(2)(0-2) \istbA(1.0){a}[al]{(1,-1)} \istb{\bar{a}}[ar]{(-1,1)} \endist
  \xtdistance{10mm}{10mm}
  \istroot(3)(1-1) \istbA(1.0){a}[al]{(-1,1)} \istb{\bar{a}}[ar]{(1,-1)} \endist
  \setxtinfosetstyle{dashed,blue,thick}
  \xtCInfoset(0)(1)<0.2>{\textcolor{blue!70}{$I^B$}}[above]
  \setxtinfosetstyle{dashed,red,thick}
  \xtCInfoset(1-1)(2)<1.4>{\textcolor{red!70}{$I^A$}}[above]
\end{istgame}
}
  \caption{}
  \label{fig:noNE:e}
\end{subfigure}
\begin{subfigure}[b]{0.28\linewidth}
\centering
\scalebox{.8}{$
\begin{array}{r|c|c|c|c}
  \multicolumn{1}{c}{\rule[-1.2ex]{0pt}{3ex}}						&	\multicolumn{1}{c}{\mathwordbox{a}{xxxxx}}		&	\multicolumn{1}{c}{\mathwordbox{\bar{a}}{xxxxx}}	\\\cline{2-3}
  \rule[-3ex]{0pt}{8ex}b			&	-1,1									&	1,-1						\\\cline{2-3}
  \rule[-3ex]{0pt}{8ex}\bar{b}			&	1,-1								&	-1,1						\\\cline{2-3}
\end{array}
$}
\caption{}
\label{fig:noNE:f}
\end{subfigure}
\caption{The mechanised graphs for forgetful Bob (Example \ref{ex:forget}) using (a) behavioural or (b) mixed policies, with normal-form in (c). (d) The mechanised graph for absent-minded Bob (Example \ref{ex:absentminded}) using a behavioural policy, with EFG and normal-form representations in (e) and (f). 
}
\label{fig:noNE}
\end{figure}

Examples \ref{ex:forget} and \ref{ex:absentminded} demonstrate that both types of imperfect recall can mean an NE in behavioural policies may not exist, even in zero-sum two agent MAIDs with binary decisions. The normal-form games (in Figures \ref{fig:noNE:c} and \ref{fig:noNE:f}) show that neither contains an NE in pure policies. It is also easy to prove non-existence in behavioural policies (see Appendix~\ref{app:proofs}). This arises due to the grand best response function being non-convex valued, which violates a condition of Kakutani's fixed point theorem.

\begin{proposition}
  \label{prop:noNE}
  Both forgetfulness and absent-mindedness can prevent the existence of an NE in behavioural policies.
\end{proposition}

\section{Solution Concepts for MAIDs under Imperfect Recall}
\label{sec:imperfect}

To overcome the fact that a behavioural policy NE may not exist in imperfect recall MAIDs, one can use mixed or correlated policies. These ensure that the grand best response function always satisfies the conditions of Kakutani's fixed point theorem, so an equilibrium always exists. We show how the assumptions behind mixed policies, behavioural mixtures, and correlated equilibria (well-studied in EFGs \citep{kaneko1995behavior,von2008extensive}, but unexplored in MAIDs) are made graphically explicit in mechanised graphs. 

\subsection{Mixed Policies and Behavioural Mixtures}
\label{sec:mixed}

 Behavioural policies allow agents to randomise independently at every decision node. By contrast, a \textbf{mixed policy} $\mu^i \in \Delta ({\dot{\bm{P}}}^i)$ is a distribution over pure policies. It allows an agent to coordinate their choice of decision rules at different decisions by randomising once at the game's outset and then committing to the assigned pure policy. More generally, \textbf{behavioural mixtures} in $\Delta ({\bm{P}}^i)$ are distributions over all behavioural policies. They
 allow agents to randomise \textit{both} at the outset of the game and before each decision. The outcome of the first randomisation determines the distributions for the others. 

A behavioural mixture changes the specification of the game because it can require correlation between different decision rules. At the object-level, a behavioural mixture for agent $i$ requires a new (correlation) decision variable $C^i$ with $\Pa_{C^i} = \varnothing$, $\Ch_{C^i} = \bm{D}^i$, and $\dom(C^i) = {\bm{P}}^i$ (the set of all behavioural policies).
The decision rules for each $D^i$ become conditional on $C^i$, so each value of $C^i$ determines a behavioural policy. This explains why $C^i$ and still every $D \in \bm{D}^i$ are decision nodes -- the agent chooses the CPDs for both. Even in the mixed policy case, where each $D^i$ depends deterministically on $C^i$, the agent chooses the dependence independently from choosing the distribution over $C^i$.  In the mechanised graph (see Figure \ref{fig:absent-minded:c}), $C^i$ gets an associated mechanism variable $\Pi_{C^i}$ for the distribution $C^i$ is drawing from (its mechanism parents are again determined by $s$-reachability). 

In EFGs, the mechanism by which agents decide on their decision rules is not explicitly shown. Mechanised graphs, however, show clearly when an agent chooses to randomise. Behavioural and mixed policies are the limiting cases of behavioural mixtures: the former where the distribution over ${\bm{P}}^i$ is deterministic; the latter where the decision rules $\bm{\Pi}_{\bm{D}^i}$ are deterministic. The difference between forgetful Bob in Example~\ref{ex:forget} using a behavioural or mixed policy is shown in Figures \ref{fig:noNE:a} and \ref{fig:noNE:b}. For Bob's behavioural policy, $C^B$ and $\Pi_{C^B}$ are omitted as the decision rules $\Pi_{B_1}$ and $\Pi_{B_2}$ are independent. This leaves a normal mechanised graph. Whereas, if Bob uses a mixed policy, he only randomises once from $\Pi_{C^B}$ at the start of the game to select a pure policy at $C^B$. This fixes deterministic decision rules at $\dot{\Pi}_{B_1}$ and $\dot{\Pi}_{B_2}$.

\begin{proposition}
  \label{prop:BRdeviation}
Given a MAID $\model$ with any partial profile ${\bm{\pi}}^{-i}$ for agents $-i$, then if agent $i$ is not absent-minded, for any behavioural policy ${\bm{\pi}}^i$ there exists a pure policy ${\dot{\bm{\pi}}}^i$ which yields a payoff at least as high against ${\bm{\pi}}^{-i}$.
    On the other hand, if agent $i$ is absent-minded in $\model$ across a pair of decisions with descendants in $\bm{U}^i$, then there exists a parameterisation of $\model$ and a behavioural policy 
${\bm{\pi}}^i$ which yields a payoff strictly higher than any payoff achievable by a pure policy.
\end{proposition}

Proposition \ref{prop:BRdeviation} says that a non-absent-minded agent cannot achieve more expected utility by using a behavioural rather than a pure (or mixed) policy, but an absent-minded agent often can. Consider Figure \ref{fig:absent-minded:c}, where $\dom(C^D) = \dot{{\bm{P}}}^D$, the set of all the driver's pure policies. $\Pi_{C^D}$ represents the distribution over $\dom(C^D)$, so $D_1$ and $D_2$ must both be~$e$ or both be~$c$. Therefore, $EU^D \leq 1$ under any mixed policy. Whereas, under the behavioural policy $\pi^1_D(e)=\frac{1}{3}$, $EU^D = \frac{4}{3}$.
This highlights an important difference between absent-mindedness and forgetfulness. Under perfect recall, every mixed policy has an equivalent behavioural policy, in the sense of inducing the same distribution over outcomes against every opposing policy profile \cite{causalgames}. Under forgetfulness, whilst a mixed policy might not have an equivalent behavioural policy, a behavioural policy always has an equivalent mixed policy \cite{Kuhn1953}, so there must exist a pure policy which performs just as well. On the other hand, under absent-mindedness, neither mixed nor behavioural policies are guaranteed to have an equivalent of the other type, so there can be a behavioural policy which outperforms every mixed policy against a given policy profile.

We introduce mixed policies (and behavioural mixtures) to MAIDs to allow more generality in modelling when agents randomise and to guarantee an NE. However, a mixed policy can require exponentially more parameters $\mathcal{O}(2^{2^{|\bm{V}|}})$ than a behavioural policy $\mathcal{O}(2^{|\bm{V}|})$ to define. Moreover, single agents are often more naturally modelled as randomising once they meet decision points \citep{Kuhn1953} (this changes for team situations described in Section \ref{sec:markov}). 
It is therefore important to know when existence of each type of NE is guaranteed. The sufficient recall result was proved by \cite{causalgames}, which we adapt to get the sufficient information result (in Appendix \ref{app:proofs}). The mixed policies result follows directly from Nash's theorem \citep{nash1950equilibrium}.

\begin{proposition}
  \label{prop:suffexistence}
  A MAID with sufficient information always has an NE in pure policies, a MAID with sufficient recall always has an NE in behavioural policies, and every MAID has an NE in mixed policies.
\end{proposition}

    Since both sufficient recall and sufficient information (Definition \ref{def:sufficient}) can be checked in poly-time\footnote{The mechanised graph is constructed using $s$-reachability, which uses the poly-time graphical criterion d-separation \citep{shachter2013bayes}.}, they expand the class of games that have simple NEs beyond those identifiable using an EFG. For example, we can check in poly-time that the MAID in Figure \ref{fig:taxi:a} is an imperfect, but sufficient, information game, and hence know that there must exist an NE in pure policies.

\subsection{Correlated Equilibria}
\label{sec:correlated}

 We have just shown how mechanised graphs can explicitly represent the assumption behind mixed policies: a \emph{single} agent uses a source of randomness to correlate their decision rules. We now do the same for when \emph{multiple} agents can use the same source of randomness, so the choice of pure policy made by each agent may be correlated. An equilibrium in such a game is called a \emph{correlated equilibrium (CE)}~\citep{aumann1974subjectivity}, which is a distribution $\kappa$ over the set of all pure policy profiles, i.e., $\kappa \in \Delta(\dot{\bm{P}})$. A mediator samples $\dot{\bm{\pi}}$ according to $\kappa$, then recommends to each agent $i$ the pure policy $\dot{\bm{\pi}}^i$. 
  The distribution $\kappa$ is a CE if no agent, given their information, has an incentive to unilaterally deviate from their recommended policy~$\dot{\bm{\pi}}^i$.

\begin{definition}
    \label{def:CE}
    In a MAID, $\kappa \in \Delta(\dot{\bm{P}})$ is a \textbf{correlated equilibrium (CE)} if and only if $ \forall i$, $\forall \dot{{\bm{\pi}}}^i, \dot{{\bm{\varpi}}}^i \in \dot{{\bm{P}}}^i$:
  \[
  \sum_{\dot{\bm{\pi}}^{-i} \in \dot{\bm{P}}^{-i}} 
  \kappa({\bm{\dot{\pi}}}^i, {\bm{\dot{\pi}}}^{-i})EU^i({\bm{\dot{\pi}}}^i, {\bm{\dot{\pi}}}^{-i}) 
  \geq 
  \sum_{\dot{\bm{\pi}}^{-i} \in \dot{\bm{P}}^{-i}}
  \kappa({\bm{\dot{\pi}}}^i, {\bm{\dot{\pi}}}^{-i}) EU^i(\dot{\bm{\pi}}^{-i}, \dot{{\bm{\varpi}}}^i)
  \]
  \end{definition}

  \begin{figure}[tb]
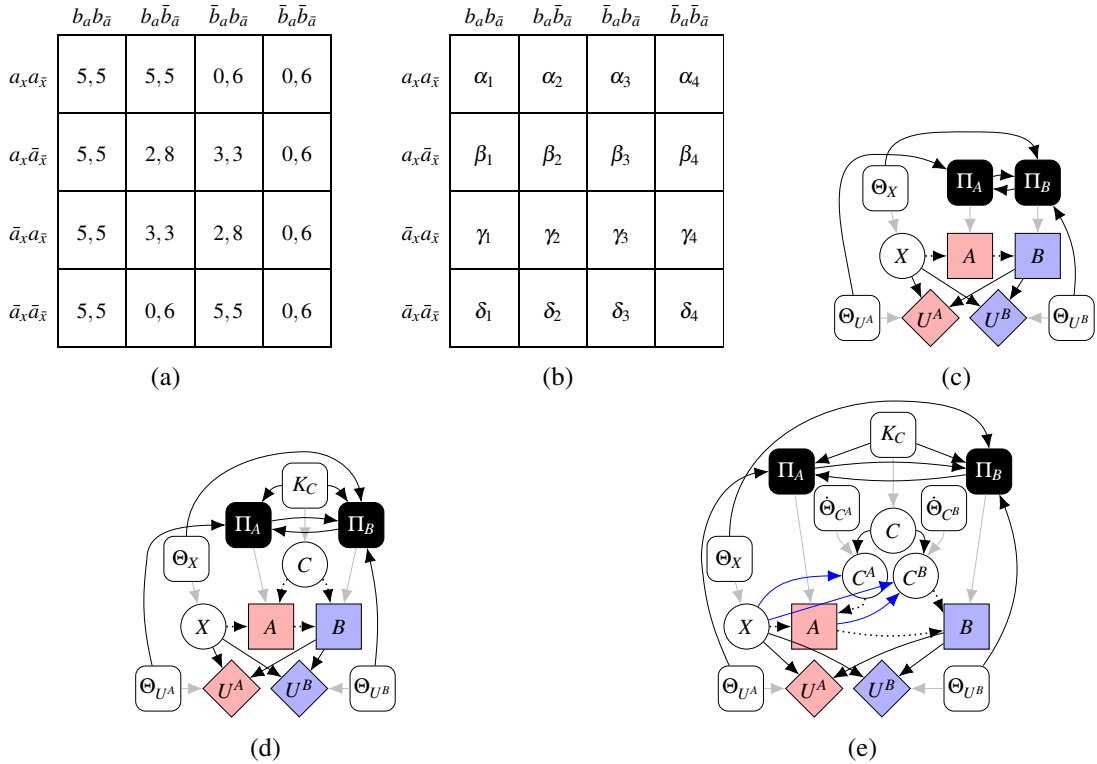

    \centering
 \begin{subfigure}[b]{0.32\linewidth}
      \vspace{-1em}
  \centering
    \scalebox{0.75}{$
      \begin{array}{r|c|c|c|c|c|c}
        \multicolumn{1}{c}{\rule[-1.2ex]{0pt}{3ex}}						&	\multicolumn{1}{c}{\mathwordbox{b_ab_{\bar{a}}}{xxxxx}}	&	\multicolumn{1}{c}{\mathwordbox{b_a \bar{b}_{\bar{a}}}{xxxxx}}		&	\multicolumn{1}{c}{\mathwordbox{\bar{b}_a b_{\bar{a}}}{xxxxx}}	&	\multicolumn{1}{c}{\mathwordbox{\bar{b}_a \bar{b}_{\bar{a}}}{xxxxx}}	\\\cline{2-5}
        \rule[-3ex]{0pt}{8ex}a_x a_{\bar{x}}			&	5,5					&	5,5							&	0,6				&		0,6					\\\cline{2-5}
        \rule[-3ex]{0pt}{8ex}a_x \bar{a}_{\bar{x}}		&	5,5					&	2,8						&	3,3				&	0,6						\\\cline{2-5}
        \rule[-3ex]{0pt}{8ex}\bar{a}_x	a_{\bar{x}}		&	5,5					&	3,3						&	2,8				&	0,6						\\\cline{2-5}
        \rule[-3ex]{0pt}{8ex}\bar{a}_x	\bar{a}_{\bar{x}}		&	5,5					&	0,6						&	5,5				&	0,6						\\\cline{2-5}
      \end{array}
      $}
    \caption{}
    \label{fig:CE:a}
  \end{subfigure}
\begin{subfigure}[b]{0.32\linewidth}
      \vspace{-1em}
  \centering
    \scalebox{0.75}{$
        \begin{array}{r|c|c|c|c|c|c}
          \multicolumn{1}{c}{\rule[-1.2ex]{0pt}{3ex}}						&	\multicolumn{1}{c}{\mathwordbox{b_ab_{\bar{a}}}{xxxxx}}	&	\multicolumn{1}{c}{\mathwordbox{b_a \bar{b}_{\bar{a}}}{xxxxx}}		&	\multicolumn{1}{c}{\mathwordbox{\bar{b}_a b_{\bar{a}}}{xxxxx}}	&	\multicolumn{1}{c}{\mathwordbox{\bar{b}_a \bar{b}_{\bar{a}}}{xxxxx}}	\\\cline{2-5}
          \rule[-3ex]{0pt}{8ex}a_x a_{\bar{x}}			&	\alpha_1					&		\alpha_2						&	\alpha_3				&		\alpha_4					\\\cline{2-5}
          \rule[-3ex]{0pt}{8ex}a_x \bar{a}_{\bar{x}}		&	\beta_1					&	\beta_2						&	\beta_3				&	\beta_4						\\\cline{2-5}
          \rule[-3ex]{0pt}{8ex}\bar{a}_x	a_{\bar{x}}		&	\gamma_1					&	\gamma_2						&	\gamma_3				&	\gamma_4						\\\cline{2-5}
          \rule[-3ex]{0pt}{8ex}\bar{a}_x	\bar{a}_{\bar{x}}		&	\delta_1					&	\delta_2						&	\delta_3				&	\delta_4						\\\cline{2-5}
        \end{array}
        $}
    \caption{}
    \label{fig:CE:b}
  \end{subfigure}
    \begin{subfigure}[b]{0.33\linewidth}
    \vspace{-1em}
  \centering
  \resizebox{0.75\width}{!}{
    \begin{influence-diagram}
      \node (X) [] {$X$};
      \node (D1) [decision, right = 1.2cm of X, player1] {$A$};
      \node (D2) [decision, right = 1.2cm of D1, player2] {$B$};
      \node (D1mec) [relevanceb, above = 1.3cm of D1] {$\Pi_A$};
      \node (D2mec) [relevanceb, above = 1.3cm of D2] {$\Pi_B$};

      \node (Xmec) [relevancew, above left = 1.2cm and 0.3cm of X] {$\Theta_X$};
      \node (U1) [utility, below right = 1.1cm and 0.5cm of X, player1] {$U^A$};
      \node (U2) [utility, below right = 1.1cm and 0.5cm of D1, player2] {$U^B$};
      
      \edge {X} {U1, U2};
      \edge [information] {X} {D1};
      \edge [information] {D1} {D2};
      \edge {D2} {U1, U2};
      \edge [gray!50] {Xmec} {X};
      \edge [gray!50] {D1mec} {D1};
      \edge [gray!50] {D2mec} {D2};

      \path (D1mec) edge[->, bend left=15] (D2mec);
      \path (D2mec) edge[->, bend left=15] (D1mec);

      \node (U1mec) [relevancew, left = 1.3cm of U1] {$\Theta_{U^A}$};
      \node (U2mec) [relevancew, right = 1.3cm of U2] {$\Theta_{U^B}$};
      \edge [gray!50] {U1mec} {U1};
      \edge [gray!50] {U2mec} {U2};
      \path (U2mec) edge[->, bend right=20] (D2mec);

      \node (space1) [minimum size=0mm, node distance=2mm, above = 0.9cm of D1mec, draw=none] {};
      \draw (Xmec) edge[,in=180,out=90] (space1.center)
      (space1.center) edge[,->,out=0,in=90] (D2mec);

      \node (space2) [minimum size=0mm, node distance=2mm, above left = 0.6cm and 0.2cm of Xmec, draw=none] {};
      \draw (U1mec) edge[,in=180,out=100] (space2.center)
      (space2.center) edge[,->,out=0,in=150] (D1mec);

  \end{influence-diagram}
  }
  \caption{}
  \label{fig:CE:c}
\end{subfigure}

     \begin{subfigure}[b]{0.49\linewidth}
    \centering
    \resizebox{0.75\width}{!}{
      \begin{influence-diagram}
        \node (X) [] {$X$};
        \node (D1) [decision, right = 1.2cm of X, player1] {$A$};
        \node (D2) [decision, right = 1.2cm of D1, player2] {$B$};
        \node (C) [above right = 1.1cm and 0.6cm of D1] {$C$};
        \node (Cmec) [relevancew, above = 1.4cm of C] {$K_C$};
        \node (D1mec) [relevanceb, above left = 1.8cm and 0.4cm of D1] {$\Pi_A$};
        \node (D2mec) [relevanceb, above right = 1.8cm and 0.4cm of D2] {$\Pi_B$};
    
        \node (Xmec) [relevancew, above left = 1.2cm and 0.3cm of X] {$\Theta_X$};
        \node (U1) [utility, below right = 1.1cm and 0.5cm of X, player1] {$U^A$};
        \node (U2) [utility, below right = 1.1cm and 0.5cm of D1, player2] {$U^B$};
        
        \edge {X} {U1, U2};
        \edge [information] {X} {D1};
        \edge [information] {D1} {D2};
        \edge {D2} {U1, U2};
        \edge [gray!50] {Xmec} {X};
        \edge [gray!50] {D1mec} {D1};
        \edge [gray!50] {D2mec} {D2};
        
        \edge [gray!50] {Cmec} {C};

        \draw (C) edge[,->,thick,dotted,in=110,out=-40] (D2);
        \draw (C) edge[,->,thick,dotted,in=70,out=-140] (D1);

        \draw (Cmec) edge[,->,in=120,out=0] (D2mec);
        \draw (Cmec) edge[,->,in=60,out=180] (D1mec);

        \path (D1mec) edge[->, bend left=10] (D2mec);
        \path (D2mec) edge[->, bend left=10] (D1mec);
    
        \node (U1mec) [relevancew, left = 1.3cm of U1] {$\Theta_{U^A}$};
        \node (U2mec) [relevancew, right = 1.3cm of U2] {$\Theta_{U^B}$};
        \edge [gray!50] {U1mec} {U1};
        \edge [gray!50] {U2mec} {U2};
        \path (U2mec) edge[->, bend right=10] (D2mec);
    
        \node (space1) [minimum size=0mm, node distance=2mm, above = 0.6cm of Cmec, draw=none] {};
        \draw (Xmec) edge[,in=180,out=90] (space1.center)
        (space1.center) edge[,->,out=0,in=90] (D2mec);
    
        \node (space2) [minimum size=0mm, node distance=2mm, left = 0.7cm of Xmec, draw=none] {};
        \draw (U1mec) edge[,in=-90,out=90, bend left=10] (space2.center)
        (space2.center) edge[,->,out=90,in=180] (D1mec);

    \end{influence-diagram}
    }
    \caption{}
    \label{fig:CE:d}
    \end{subfigure}
    \begin{subfigure}[b]{0.49\linewidth}
      \centering
      \resizebox{0.75\width}{!}{
        \begin{influence-diagram}
          \node (X) [] {$X$};
          \node (D1) [decision, right = 1.2cm of X, player1] {$A$};
          \node (D2) [decision, right = 2.7cm of D1, player2] {$B$};
          \node (C) [above right = 1.7cm and 1.4cm of D1] {$C$};
          \node (Cmec) [relevancew, above = 1.7cm of C] {$K_C$};
          \node (D1mec) [relevanceb, above left = 2.75cm and 0.4cm of D1] {$\Pi_A$};
          \node (D2mec) [relevanceb, above right = 2.75cm and 0.4cm of D2] {$\Pi_B$};
  
          \node (Xmec) [relevancew, above left = 1.2cm and 0.3cm of X] {$\Theta_X$};
          \node (U1) [utility, below right = 1.1cm and 1.2cm of X, player1] {$U^A$};
          \node (U2) [utility, right = 1.2cm of U1, player2] {$U^B$};

          \node (CA) [above right = 0.9cm and 0.9cm of D1] {$C^A$};
          \node (CB) [above left = 0.9cm and 0.9cm of D2] {$C^B$};

          \node (CAmec) [relevancew, above left = 1.2cm and 0.5cm of CA] {$\dot{\Theta}_{C^A}$};
          \node (CBmec) [relevancew, above right = 1.2cm and 0.5cm of CB] {$\dot{\Theta}_{C^B}$};

          \edge {X} {U1};
          \edge [information] {X} {D1};
          \edge {D2} {U2};
          \edge [gray!50] {Xmec} {X};
          \edge [gray!50] {D1mec} {D1};
          \edge [gray!50] {D2mec} {D2};
          \edge [gray!50] {Cmec} {C};
          \edge {Cmec} {D1mec, D2mec};

          \path (D1mec) edge[->, bend left=8] (D2mec);
          \path (D2mec) edge[->, bend left=8] (D1mec);
  
          \node (U1mec) [relevancew, left = 1.3cm of U1] {$\Theta_{U^A}$};
          \node (U2mec) [relevancew, right = 1.5cm of U2] {$\Theta_{U^B}$};
          \edge [gray!50] {U1mec} {U1};
          \edge [gray!50] {U2mec} {U2};
          \path (U2mec) edge[->, bend right=30] (D2mec);
  
          \path (D2) edge[->, bend right=7] (U1);
          \path (X) edge[->, bend left=7] (U2);

          \draw (CB) edge[,->,thick,dotted,in=150,out=-40] (D2);
          \draw (CA) edge[,->,thick,dotted,in=30,out=-90] (D1);

          \draw [gray!50] (CBmec) edge[,->,in=60,out=-90] (CB);
          \draw [gray!50] (CAmec) edge[,->,in=120,out=-90] (CA);

          \draw (X) edge[,->,in=180,out=60, blue] (CA);
          \edge [blue] {X} {CB};
          \path (D1) edge[->, bend right=20, blue] (CB);

          \draw (C) edge[,->,in=110,out=-180] (CA);
          \draw (C) edge[,->,in=70,out=0] (CB);

          \draw (D1) edge[,->,thick,dotted,bend right = 10] (D2);

          \node (space1) [minimum size=0mm, node distance=2mm, above = 0.6cm of Cmec, draw=none] {};
          \draw (Xmec) edge[,in=180,out=90] (space1.center)
          (space1.center) edge[,->,out=0,in=90] (D2mec);
  
          \node (space2) [minimum size=0mm, node distance=2mm, left = 0.5cm of Xmec, draw=none] {};
          \draw (U1mec) edge[,in=-90,out=90, bend left=10] (space2.center)
          (space2.center) edge[,->,out=90,in=180] (D1mec);

      \end{influence-diagram}
      }
      \caption{}
      \label{fig:CE:e}
  \end{subfigure}

      \caption{The sub-figures (a) and (b) give the expected payoff for each agent under each pure policy profile and the parameterisation of the distribution $\kappa$, respectively. The mechanised graph for Example~\ref{ex:signaling}'s original MAID is shown in (c), and the mechanised graphs for when a trusted mediator gives public or private recommendations to find a CE are shown in (d) and (e), respectively. The blue edges are added to the graph in (e) for a MAID-CE's staggered recommendations.
            }
    \label{fig:CE}
  \end{figure}

We illustrate how MAIDs and their mechanised graphs make explicit the assumptions used for a CE using a costless-signal variation of Spence's job market game \cite{spence1978job}. 

\begin{example}
  \label{ex:signaling}
  Alice is hardworking or lazy ($X$) with equal probability. She applies for a job with Bob by deciding which costless signal ($A$) to send. Bob can distinguish between the signals, but does not know Alice's true temperament. He decides whether to offer the job ($B$) to Alice. The utility functions for Alice and Bob are $U^A = (6-2X)\cdot B$ and $U^B = 6 + (10X-6)\cdot B$, respectively.
\end{example}

The mechanised graph for the original game's MAID is shown in Figure \ref{fig:CE:c}. The cycle between $\Pi_A$ and $\Pi_B$ reveals that each agent’s decision rule strategically relies on the other agent’s decision rule.\footnote{That Bob strategically relies on Alice’s decision rule might be less obvious than the fact that Alice strategically relies on Bob’s decision rule. The dependency occurs because since Bob can observe $A$, this unblocks an active path $\Pi_A \rightarrow A \leftarrow X \rightarrow U^B$ in the independent mechanised graph, so $\Pi_A$ is $s$-reachable from $\Pi_B$.} Therefore, the MAID has insufficient information and no proper subgames, making it difficult to solve. 

To find the CE of this game, a trusted mediator is added using a \emph{correlation variable} $C$ with $\Pa_{C} = \varnothing$, $\Ch_{C} = \bm{D}$, and $\dom(C) = \dot{\bm{P}}$. In the mechanised graph, $C$'s associated mechanism variable $K_{C}$ represents the distribution $\kappa \in \Delta(\dot{\bm{P}})$ that the mediator draws a pure policy profile according to. This time, since $K_{C}$ is fixed as $\kappa$ at the game outset instead of being chosen by any agent, $C$ acts as a chance variable (in contrast to the correlation decision variable introduced for mixed policies and behavioural mixtures).

There is a well-known difference between public and private recommendations. If public, every payoff in the convex hull of the set of NE payoffs can be attained by a CE; however, if the recommendations are private, then the payoffs to each agent in a CE can lie outside this convex hull (e.g., Aumann's game of chicken \cite{aumann1974subjectivity}). This distinction is made explicit in the MAID's graph. If the recommendations are public, then the full outcome of $C$ (the pure policy profile chosen by the mediator) is known by every agent (shown by the dotted edges between $C$ and both $A$ and $B$ in Figure \ref{fig:CE:d}). If the recommendations are private, then each agent only observes their decision rules (action recommendations) in $C$'s outcome, i.e., all recommendations given to other players are hidden (at $C^A$ and $C^B$ in Figure \ref{fig:CE:e}). In this latter case, the agent infers, using Bayes' rule, a posterior over the pure policy profile that was chosen (and also which action was recommended to the other agent(s)). If $\kappa$ is a CE, then each agent picks for their decision $D$'s decision rule the mediator's recommendation, i.e., $\dot{\pi}_D$ where $c = \dot{\bm{\pi}}$. 
The set of variables $\bm{D}$ remain as decisions because agents are free to deviate from their recommendation and pick any CPDs as decision rules for their decisions. 

This mediator's distribution  $\kappa  \in \Delta(\dot{\bm{P}})$ can be parameterised according to that in Figure \ref{fig:CE:b}. Note that $b_a\bar{b}_{\bar{a}}$ denotes the pure policy profile where Bob offers the job ($b$) to Alice if she selects $a$ and Bob does not offer the job ($\bar{b}$) if Alice selects $\bar{a}$. Using the expected payoff for Alice and Bob under each pure policy profile (Figure \ref{fig:CE:a}),  Definition~\ref{def:CE}'s incentive constraints define 24 inequalities that must be satisfied by the CE distribution. After some algebra, we find that  $\alpha_1 = \alpha_2 = \alpha_3 = \beta_1 = \beta_2 = \beta_3 = \gamma_1 = \gamma_2 = \gamma_3 = 0$; $\alpha_4, \beta_4, \gamma_4, \delta_4 \geq 0$; $\alpha_4 -2\beta_4 + 3\gamma_4 \geq 0$, and $3\beta_4 -2\gamma_4 + \delta_4 \geq 0$. Any CE, therefore, has Bob never offering a job to Alice because they play the pure policy $\bar{b}_a\bar{b}_{\bar{a}}$ with probability 1, i.e., Bob's decision rule has $\pi^B(B=\bar{b}\mid A=a)=\pi^B(B=\bar{b}\mid A=\bar{a})=1$. The remaining constraints require Alice not to give any incentive for Bob to offer her a job by making the conditional probability of Alice being hardworking too high relative to the conditional probability of her being lazy when he receives the signal $a$ or $\bar{a}$. These constraints find that every CE will result in $EU^A = 0$ and $EU^B = 6$. This is unsurprising because, in a signaling game with costless signals, every CE will be a `pooling equilibrium' \cite{cho1987signaling} (an equilibrium in which Alice chooses the same action regardless of their temperament). 
  
  Whilst the CE is among the best-known solution concepts for normal-form games, and is efficiently computable in that setting (e.g., via linear programming \citep{hart1989existence}), there can be an exponential number of pure policies (so an exponential number of incentive constraints) in EFGs and even in bounded treewidth MAIDs. 
  It is therefore currently unknown if a CE can be found in an EFG or MAID in poly-time. Motivated by these tractability concerns, Von Stengel and Forges proposed an \emph{extensive-form correlated equilibrium (EFCE)} \cite{von2008extensive}. Along similar lines, we define a \emph{MAID correlated equilibrium}. 
  
   Instead of revealing the entire recommendation $\dot{\bm{\pi}}^i$ to each agent $i$ immediately, we let the mediator \emph{stagger} their recommendations. This is made visible in the mechanised graph by adding the blue edges in Figure \ref{fig:CE:e}. Importantly, if an agent deviates from any recommendation, then the mediator will \emph{cease giving further recommendations to that agent} (but will still give recommendations to all other agents). Thus, the incentive constraints are now tied to the threat of the mediator withholding future information.

   \begin{definition}
    \label{def:MAIDCE}
    Given a distribution $\kappa \in \Delta(\dot{\bm{P}})$, consider the MAID with an additional correlation variable $C$ with $\Pa_{C} = \varnothing$, $\Ch_{C} = \{C_D\}_{D \in \bm{D}}$, and $\Ch_{C_D} = \{D\}$ for each $D$. Let a pure policy profile $\bm{\dot{\pi}}$ be selected at $C$ according to $\kappa$. Then, when each decision context $\pa_D$ is reached, agent $i$ receives a recommended move $d \in \dom(D)$ specified by $\bm{\dot{\pi}}_D \in \bm{\dot{\pi}}$ ($C_D$ hides all other recommendations $\bm{\dot{\pi}}_{-D} \in \bm{\dot{\pi}}$). A \textbf{MAID correlated equilibrium (MAID-CE)} is an NE of this game in which no agent has an incentive to deviate from their recommendations.
  \end{definition}

 The localised recommendations in a MAID-CE pose weaker incentive constraints compared to a CE, so the set of MAID-CE outcomes is larger. As such, MAID-CEs can lead to Pareto-improvements over the CEs (and NEs) in a game. We now give one such MAID-CE. The mediator chooses a signal $s$ with equal probability for type $X=x$, i.e., $\Pr(c_A=a\mid X=x)= \Pr(c_A=\bar{a}\mid X=x)=0.5$. Bob is recommended to offer Alice a job ($b$) when Alice's action matches $s$ and to reject otherwise ($\bar{b}$). If $X=\bar{x}$, then the recommendation to Alice is arbitrary and is independent of the signal $s$, which is only shown to hardworking Alice. Because the mediator only gives Alice her recommendation once her decision context $\Pa_A$ is set, lazy Alice cannot know $s$. Therefore, in any situation, lazy Alice's action will match $s$ with probability $\frac{1}{2}$. Consequently, when Bob is called to play (i.e., the decision context $\Pa_B$ is set), and Alice's action matches $s$, Alice is twice as likely to be hardworking than lazy (so $EU^B = \frac{20}{3}$ for offering Alice a job rather than $EU^B = 6$ for rejecting her). If instead, Alice's action does not match $s$, then he knows with certainty that Alice is lazy, so his best response is to reject. Overall, Alice's expected payoff in this MAID-CE is $3.5$, and Bob's is $6.5$ (higher than 0 and 6, respectively, for all CEs).

A MAID-CE can be computed in poly-time if the treewidth is bounded, via a reduction to a linear program. We follow Huang et al \cite{huang2008computing}'s method because the information sets in an EFG are in bijection with the decision contexts in a MAID, but relax beyond their conditions as MAIDs only require sufficient (rather than perfect) recall \cite{huang2008computing}. 
Any distribution over pure policies induced by an NE can be represented using a distribution $\kappa$, and hence any mixed NE (or equivalent behavioural NE) is also a CE and MAID-CE.
As every MAID has an NE in (mixed) policies, every MAID must also have a CE and a MAID-CE.

  \begin{proposition}
    \label{prop:CEpoly}
    A MAID-CE in bounded treewidth MAIDs with sufficient recall can be found in~poly-time.
  \end{proposition}

  \section{Complexity Results in MAIDs}
  \label{sec:complexity}
  
 We now give some complexity results in MAIDs. Our first follows from the known result in normal-form games \citep{daskalakis2009complexity}. Any normal-form game $\mathcal{N}$ can be reduced to a MAID where each agent has one utility node (which copies the payoffs in $\mathcal{N}$) and one decision node. The domains of the decision variables are the set of each agent's pure strategies in $\mathcal{N}$. Edges are added from every $D \in \bm{D}$ to every $U \in \bm{U}$.
 
 \begin{proposition}
     In a MAID, finding an NE in mixed policies is \PPAD-hard.
 \end{proposition}

\begin{table}[h]
    \centering
    \begin{tabular}{l l l}
        \toprule
        \textbf{Problem} & \textbf{Input} & \textbf{Question}\\
        \midrule
        \textsc{Is-Best-Response} & $\model$, $i$, ${\bm{\pi}}^{-i}$, $q \in \Rats$ & Is there some $\hat{{\bm{\pi}}}^i$ such that $EU^i(\hat{{\bm{\pi}}}^i, {\bm{\pi}}^{-i}) > q$?\\
        \textsc{Is-Nash} & $\model$, ${\bm{\pi}}$ & Is ${\bm{\pi}}$ a (behavioural) NE of $\model$?\\
        \textsc{Non-Emptiness}: & $\model$ & Does $\model$ have a (behavioural) NE?\\
        \bottomrule
    \end{tabular}
    \caption{Three decision problems in MAIDs with behavioural policies.}
    \label{tab:problems}
\end{table}

In the following results, we focus on the complexity of the decision problems in Table \ref{tab:problems}.

  \begin{proposition}
    \label{prop:decidebestresponse}
    \textsc{Is-Best-Response} is ${\NP^{\PP}}$-complete, \NP-complete when restricted to MAIDs with graphs of bounded treewidth, and \PP-complete if both $|\bm{D}^i|$ and the in-degrees of $\bm{D}^i$ are bounded.
  \end{proposition}

  \begin{proof}[Proof sketch]
      \textsc{Is-Best-Response} is in ${\NP^{\PP}}$ because given $\hat{{\bm{\pi}}}^i$, we can verify that $EU^i(\hat{{\bm{\pi}}}^i, {\bm{\pi}}^{-i}) > q$ in poly-time using a \PP~oracle for inference in a BN \cite{littman2001stochastic}. With bounded treewidth, verification can be done in poly-time. The final setting is in \PP~by analogy with Kwisthout's \textsc{PARAMETER TUNING} \cite{kwisthout2012computational}. For the general case's hardness, we can reduce from \textsc{E-Majsat} as in \citep{park2004complexity}, where MAP-nodes are replaced by agent $i$'s decision nodes; for bounded treewidth, we can reduce from \textsc{MAXSAT} as in \citep{de2005inferential}; and for the final case, \textsc{Is-Best-Response} with $|\bm{D}^i| = 0$ is the same as inference in a BN.
  \end{proof}

Proposition \ref{prop:decidebestresponse} suggests \textsc{Is-Best-Response} is, in general, only tractable if inference is easy \emph{and} $|\bm{D}^i|$ is bounded by a constant. Proposition \ref{prop:optimalbestresponse} then explains the decision problem's name.

\begin{proposition}
\label{prop:optimalbestresponse}
If the in-degrees of $\bm{D}^i$ are bounded and \textsc{Is-Best-Response} can be solved in poly-time, then a best response policy for agent $i$ to a partial profile ${\bm{\pi}}^{-i}$ can be found in polynomial~time.
\end{proposition}

\begin{proposition}
    \label{prop:isNASH}
    \textsc{Is-Nash} is ${\coNP^{\PP}}$-complete, and \coNP-complete when restricted to MAIDs with graphs of bounded treewidth. The general problem remains ${\coNP^{\PP}}$-hard in sufficient information MAIDs. In MAIDs without chance variables, the problem remains \coNP-hard. 
  \end{proposition}
\begin{proof}[Proof sketch]
  For membership, we can check that ${\bm{\pi}}$ is \emph{not} an NE by guessing an agent $i$ and checking if ${\bm{\pi^i}} \in {\bm{\pi}}$ is a best response in poly-time using a $\PP$-oracle (this is unnecessary if the graph has bounded treewidth). Hardness comes from the single-agent setting where it is the complement of \textsc{Is-Best-Response}. In MAIDs without chance variables, we reduce from partial order games \citep{zahoransky2021partial}.
\end{proof}

Proposition \ref{prop:suffexistence} shows when \textsc{Non-Emptiness} is vacuous. However, in an insufficient recall MAID, \textsc{Non-Emptiness} is, in general, intractable even without chance variables. 
  
  \begin{proposition}
    \label{prop:nonempty} \textsc{Non-Emptiness} is \NEXPTIME-hard and becomes \NEXPTIME-complete if we restrict to MAIDs without chance variables.
  \end{proposition}

\begin{proof}[Proof sketch]
  For hardness, we can reduce from partial order games. Without chance variables, we can determine \textsc{Non-Emptiness} using a similar algorithm to that in \cite{zahoransky2021partial}. It exploits the setting's determinism: payoffs are poly-time computable and the number of policy profiles is reduced to~$\mathcal{O}(2^{|\bm{V}|})$. 
\end{proof}

  \begin{proposition}
    \label{prop:polySPE}
    In a MAID with sufficient information, if the in-degrees of $\bm{D}$ are bounded and \textsc{Is-Best-Response} can be solved in poly-time, then a pure NE can be found in poly-time.
  \end{proposition}

This result suggests an NE can be found efficiently in certain MAIDs, but even in games without sufficient information, NEs can be found more efficiently in a MAID than in an EFG. The mechanised graph dependencies reveal more `subgames' -- parts of the MAID that can be solved independently from the rest -- to which dynamic programming can be applied \cite{koller2003multi,Hammond2021}. As finding an NE in both EFGs and MAIDs depends significantly on the game's size, this can empirically lead to large compute savings~\citep{koller2003multi}. 
  \section{Applications and Conclusion}
    \label{sec:markov}

 \begin{figure}[tb]
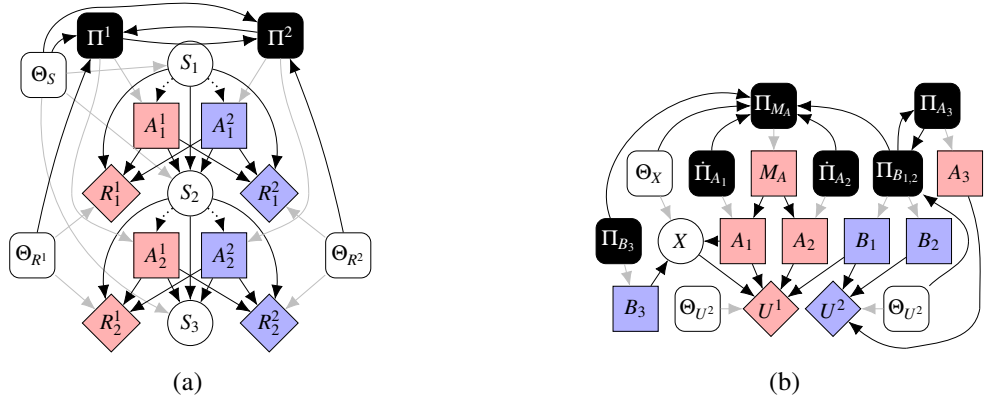

    \centering
\begin{subfigure}[b]{0.49\linewidth}
          \vspace{-1em}
        \centering
        \resizebox{0.75\width}{!}{
        \begin{influence-diagram}
        \node (S1) [] {$S_1$};
        \node (S2) [below=2.3cm of S1] {$S_2$};
        \node (D11) [decision, below left = 1.1cm and 0.6cm of S1, player1] {$A^1_1$};
        \node (D21) [decision, below right = 1.1cm and 0.6cm of S1, player2] {$A^2_1$};
    
        \node (U11) [utility, below left = 2.3cm and 1.4cm of S1, player1] {$R^1_1$};
        \node (U21) [utility, below right = 2.3cm and 1.4cm of S1, player2] {$R^2_1$};
              
        \node (D12) [decision, below left = 1.1cm and 0.6cm of S2, player1] {$A^1_2$};
        \node (D22) [decision, below right = 1.1cm and 0.6cm of S2, player2] {$A^2_2$};
        \node (U12) [utility, below left = 2.3cm and 1.4cm of S2, player1] {$R^1_2$};
       \node (U22) [utility, below right = 2.3cm and 1.4cm of S2, player2] {$R^2_2$};
  
        \node (S3) [below=2.3cm of S2] {$S_3$};

        \node (P1) [relevanceb, above left=1.6cm and 1cm of D11] {$\Pi^{1}$};
        \node (P2) [relevanceb, above right=1.6cm and 1cm of D21] {$\Pi^{2}$};

        \node (TS) [relevancew, below left=-2.1cm and 1.2cm of U11] {$\Theta_{S}$};
        
         \node (TR1) [relevancew, below left=1.1cm and 2.8cm of S2] {$\Theta_{R^1}$};
         \node (TR2) [relevancew, below right=1.1cm and 2.8cm of S2] {$\Theta_{R^2}$};
         \path [gray!50] (P1) edge[->] (D11);
        \path [gray!50] (P2) edge[->] (D21);

        \path [gray!50] (TS) edge[->] (S1);
        \path [gray!50] (TS) edge[->] (S2);
        \path [gray!50] (TS) edge[->, , bend right=35] (S3);

         \path [gray!50] (TR1) edge[->, ] (U11);
        \path [gray!50] (TR1) edge[->, ] (U12);
        \path [gray!50] (TR2) edge[->, ] (U21);
        \path [gray!50] (TR2) edge[-> ] (U22);
  
        \path (S1) edge[->, bend right=45] (U11);
        \path (S1) edge[->, bend left=45] (U21);
        \path (S2) edge[->, bend right=45] (U12);
        \path (S2) edge[->, bend left=45] (U22);
  
        \edge {S1} {S2}
        \edge {S2} {S3}

        \draw (S1) edge[,->,thick,dotted,in=90,out=-140] (D11);
        \draw (S1) edge[,->,thick,dotted,in=90,out=-40] (D21);

        \draw (S2) edge[,->,thick,dotted,in=90,out=-140] (D12);
        \draw (S2) edge[,->,thick,dotted,in=90,out=-40] (D22);
  
        \edge{D11} {U11}
        \edge{D11} {U21}
        \edge{D21} {U11}
        \edge{D21} {U21}
  
        \edge{D12} {U12}
        \edge{D12} {U22}
        \edge{D22} {U12}
        \edge{D22} {U22}
  
        \edge{D11} {S2}
        \edge{D21} {S2}
        \edge{D12} {S3}
        \edge{D22} {S3}

        \draw (TS) edge[,->,in=180,out=70] (P1);

        \path (TR1) edge[->, bend left=5] (P1);
        \path (TR2) edge[->, bend right=5] (P2);
       
        \path (P1) edge[->, bend right=8] (P2);
        \path (P2) edge[->, bend right=8] (P1);

        \node (space1) [minimum size=0mm, node distance=2mm, left = 0.7cm of U11, draw=none] {};
      \draw [gray!50] (P1) edge[,in=90,out=-90] (space1.center)
      (space1.center) edge[,->,out=-90,in=150] (D12);

     \node (space2) [minimum size=0mm, node distance=2mm, right = 0.7cm of U21, draw=none] {};
      \draw [gray!50] (P2) edge[,in=90,out=-90] (space2.center)
      (space2.center) edge[,->,out=-90,in=30] (D22);

      \node (space3) [minimum size=0mm, node distance=2mm, above = 0.5cm of P1, draw=none] {};
      \draw (TS) edge[,in=180,out=90] (space3.center)
      (space3.center) edge[,->,out=0,in=150] (P2);
        
        \end{influence-diagram}}
        \caption{}
        \label{fig:Markov:a}
    \end{subfigure}
    \begin{subfigure}[b]{0.49\linewidth}
              \vspace{-1em}
        \centering
        \resizebox{0.75\width}{!}{
          \begin{influence-diagram}
            \node (U1) [utility, player1] {$U^1$};
            \node (U2) [utility, right = 1.1cm of U1, player2] {$U^2$};
            \node (U2mec) [right = 1.3cm of U2, relevancew] {$\Theta_{U^2}$};
            \node (U1mec) [left = 1.3cm of U1, relevancew] {$\Theta_{U^2}$};

            \node (D11) [above left = 1.2cm and 0.5cm of U1, decision, player1] {$A_1$};
            \node (D12) [right = 1.1cm of D11, decision, player1] {$A_2$};
            \node (D21) [right = 1.1cm of D12, decision, player2] {$B_1$};
            \node (D22) [right = 1.1cm of D21, decision, player2] {$B_2$};
            \node (X) [left = 1.1cm of D11,] {$X$};
            \node (Xmec) [above left = 1.2cm and 0.55cm of X, relevancew] {$\Theta_X$};
            \node (D11mec) [right = 1.1cm of Xmec, relevanceb] {$\dot{\Pi}_{A_1}$};
            \node (M) [right = 1.1cm of D11mec, decision, player1] {$M_A$};
            \node (D12mec) [right = 1.1cm of M, relevanceb] {$\dot{\Pi}_{A_2}$};            
            \node (P1) [above = 1.2cm of M, relevanceb] {$\Pi_{M_A}$};
            \node (P2) [right = 1.1cm of D12mec, relevanceb] {$\Pi_{B_{1,2}}$};
            \node (D13) [right = 1.1cm of P2, decision, player1] {$A_3$};
            \node (P3) [above right = 1.2cm and .7cm of P2, relevanceb] {$\Pi_{A_3}$};
            \node (B3) [left = 1.1cm of U1mec, decision, player2] {$B_3$};
            \node (P4) [left = 1.1cm of X, relevanceb] {$\Pi_{B_3}$};

            \edge {B3} {X};
            \edge [gray!50] {P4} {B3};
            \edge [gray!50] {Xmec} {X};
            \edge {D11} {X};
            \edge {M} {D11, D12};
            \edge [gray!50] {D11mec} {D11};
            \edge [gray!50] {D12mec} {D12};
            \edge {X,D11,D12} {U1};
            \edge {D21,D22} {U2};
            \edge [gray!50] {U1mec} {U1};

            \edge [gray!50] {U2mec} {U2};
            \edge [gray!50] {P1} {M};
            \edge [gray!50] {P2} {D21, D22};
            \edge [gray!50] {P3} {D13};
            \edge {P3} {P2};
            \edge {D21} {U1};

            \draw (D11mec) edge[,->,in=200,out=90] (P1);
            \draw (D12mec) edge[,->,in=-20,out=90] (P1);
            \draw (P4) edge[,->,in=160,out=110] (P1);

            \draw (P2) edge[->,in=0,out=110] (P1);
            \draw (Xmec) edge[->,in=180,out=70] (P1);
            \draw (P2) edge[->,in=190,out=90] (P3);

            \node (space1) [minimum size=0mm, node distance=2mm, right = 1.3cm of U2mec, draw=none] {};
            \node (space3) [minimum size=0mm, node distance=2mm, below = 0.7cm of U2mec, draw=none] {};
            \draw (D13) edge[,in=90,out=-70] (space1.center)
            (space1.center) edge[,out=-100,in=10] (space3.center)
            (space3.center) edge[,->,out=180,in=-40] (U2);

            \node (space2) [minimum size=0mm, node distance=2mm, right = 0.7cm of D22, draw=none] {};
            \draw (U2mec) edge[,in=270,out=45] (space2.center)
            (space2.center) edge[,->,out=90,in=-40] (P2);
                 
        \end{influence-diagram}
        }
        \caption{}
        \label{fig:Markov:b}
    \end{subfigure}
      \caption{Mechanised graphs for a CE with (a) public and (b) private recommendations, where the blue edges are added for a MAID-CE; (c) a Markov game;(d) a team setting with imperfect communication.}
    \label{fig:teamMarkov}
  \end{figure}

      We introduced forgetfulness and absent-mindedness as properties of individual agents (due to imperfect memory). However, imperfect recall also commonly arises in \emph{team situations}; each team consists of several agents targeting a common goal with imperfect communication. Forgetfulness or absent-mindedness occurs when an agent does not know their teammates' actions (or observations) or whether they have acted at all. Mechanised graphs represent these situations where teams often employ a mix of randomisation strategies (e.g., Figure \ref{fig:Markov:b}). For mixed policies, the random seed is chosen at the start, before the agents set out following their distinct policies. For behavioural policies, agents pick a new random seed at every decision point. Behavioural mixtures correspond to randomising at both stages.

Another application of imperfect recall in MAIDs is to \emph{Markov (or `stochastic') games} \citep{shapley1953stochastic}, in which the agents move between different states over time (e.g., Figure \ref{fig:Markov:a}). At each time step~$t$, each agent~$i$ selects an action $A^i_t$, and the game probabilistically transitions to a new state $S_{t+1}$, depending on the previous state $S_{t}$ and the actions selected, and each agent receives a payoff $R^i_t$. Each $S_{t+1}$ and $R^i_{t}$ has parents $\{S_t,A^{1}_t,\dots, A^n_t\}$ and must be identically distributed for all $t$, again represented using shared mechanism variables. Often, the agent must learn a memoryless, stationary policy $\pi^i:S\to \Delta (A^i)$, where $S$ is the set of states and $\Delta (A^i)$ the set of probability distributions over agent $i$'s actions. Hence, the agents are absent-minded (every decision $A^i_{t+1}$ of agent $i$ shares the same decision rule) and use \emph{behavioural} policies (since the action selected in each state is independently stochastic). In light of Proposition~\ref{prop:noNE}, it is therefore natural to ask whether a Markov game may not have an NE in memoryless stationary policies. It is known that infinite-horizon Markov games might not (for a counterexample see \cite{de2001quantitative}). Although infinite games lie outside of the scope of this paper, it is nonetheless insightful to note that this possible non-existence is due to absent-mindedness: if agents can choose a different decision rule at each time step, a behavioural NE is guaranteed~\citep{maskin2001Markov}.

We have shown how to handle imperfect recall in MAIDs by overcoming the potential lack of NEs in behavioural policies using mixed and correlated equilibria. EFGs leave many assumptions about how agents play games hidden, but mechanised graphs make explicit the assumptions behind imperfect recall (both forgetfulness and absent-mindedness), mixed policies, and two types of correlated equilibria. Our complexity results highlight the importance of restricting the use of MAIDs to those with a limited number of decision variables and bounded treewidth. Finally, our applications to Markov games and team situations show that imperfect recall broadens the scope of what can be modelled using MAIDs.

\paragraph{Acknowledgements} 
The authors wish to thank Ryan Carey, Tom Everitt, and Francis Rhys Ward for invaluable feedback, as well as three anonymous reviewers for their helpful comments.  Fox was supported by the EPSRC Centre for Doctoral Training in Autonomous Intelligent Machines and Systems (Reference:
EP/S024050/1), MacDermott was supported by the UKRI Centre
for Doctoral Training in Safe and Trusted Artificial Intelligence (Reference: EP/S023356/1), Hammond was supported by an EPSRC Doctoral Training Partnership studentship (Reference: 2218880), and Wooldridge was supported by a UKRI
Turing AI World Leading Researcher Fellowship (Reference: EP/W002949/1).

\bibliographystyle{eptcs}
\bibliography{refs}

\appendix

\section{Strategic Relevance and Subgames}
\label{app:stratrel}

Koller and Milch define \textbf{strategic relevance} to infer whether the choice of a decision rule can affect the optimality of another decision rule \cite{koller2003multi}. Hammond et al. extend strategic relevance to also consider whether the parameterisation of non-decision nodes can affect the decision rule's optimality \cite{causalgames}. Intuitively, a mechanism $\mecvar_V$ is strategically relevant to the decision rule $\Pi_D$ of $D \in \bm{D}^i$ if the choice of CPD at $\mecvar_V$ can affect agent $i$'s utility nodes that are downstream of $D$ (i.e., those in $\bm{U}^i\cap \Desc_D$). Formally:

 \begin{definition}[\cite{koller2003multi,causalgames}]
  \label{def:strategic-relevance} Recall that $\dom(\Pi_D)$ gives the set of possible decision rules at $\Pi_D$ for decision node $D$. 
  Given a MAID with $D \in \bm{D}^i$ and $V \neq D \in \bm{D}$, the mechanism $\mecvar_V$ for $V$ is \textbf{strategically relevant} to $\Pi_D$ if there exist two joint distributions over $\bm{V}$ parameterised by mechanisms $\mecvals$ and $\mecvals'$ respectively such that:
  \begin{itemize}
      \item $\pi_{D} \in \argmax_{\varpi_{D} \in \dom(\Pi_{D})} EU^i((\varpi_{D}, {\bm{\pi}}_{-D}) \mid \mecvals)$
      \item $\mecvals$ differs from $\mecvals'$ only at $\mecvar_V$,
      \item $\pi_{D} \notin \argmax_{\varpi_{D} \in \dom(\Pi_{D})} EU^i((\varpi_{D}, {\bm{\pi}}_{-D}) \mid \mecvals')$, and neither does any decision rule $\varpi_{D}$ that agrees with $\pi_{D}$ on all $\pa_{D}$ such that $\Pr(\pa_{D} \mid \mecvals') > 0$.
  \end{itemize}
\end{definition} 

The first two conditions say: if the decision rule $\pi_{D}$ is optimal for the MAID parameterisation (i.e., the setting of all mechanism variables) $\mecvals$, and $\Pi_D$ does not strategically rely on $\mecvar_V$, then $\pi_{D}$ must also be optimal for any other parameterisation $\mecvals'$ that differs from $\mecvals$ only at $\mecvar_V$. The third condition deals with sub-optimal decision rules in response to zero-probability decision contexts (i.e., non-credible threats).

Koller and Milch \cite{koller2003multi} also derive a graphical criterion for strategic relevance, called \emph{$s$-reachability}, which is sound (if $\mecvar_V$ is strategically-relevant to $\Pi_D$, then $\mecvar_V$ is $s$-reachable from $\Pi_D$) and complete (if $\mecvar_V$ is $s$-reachable from $\Pi_D$, then there is some parameterisation $\mecvals$ of the MAID and some policy profile $\pi$ such that $\mecvar_V$ is strategically-relevant to $\Pi_D$). This uses the \emph{independent mechanised graph} $\meczero{\graph}$, which contains a separate mechanism parent for each variable in the original MAID graph, but no edges between the mechanism variables.  
 \begin{definition}[\cite{koller2003multi}]
     \label{def:$s$-reachability}
     $\mecvar_{V}$ is $s$-reachable from $\Pi_{D}$ if $\mecvar_{V} \not\perp_{\meczero{\graph}} \bm{U}^i \cap \Desc_{D} \mid D, \Pa_{D}$.
 \end{definition}

$s$-reachability determines which inter-mechanism edges are present in the MAID's mechanised graph; $\mecvar_V \rightarrow \Pi_D$ exists in the mechanised graph if and only if $\Pi_D$ strategically relies on $\mecvar_V$. 

\begin{figure}[h!]
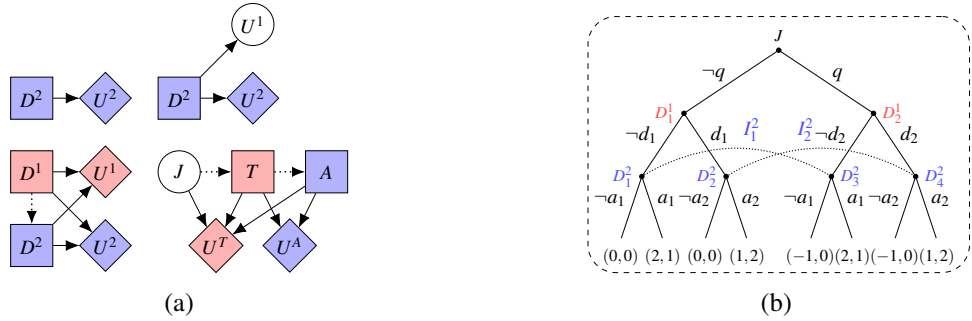

\begin{subfigure}[b]{0.49\linewidth}
  \centering
  \resizebox{0.7\width}{!}{
    \begin{influence-diagram}
      \node (D1) [decision, player1] {$D^1$};
      \node (D2) [decision, player2, below = of D1] {$D^2$};
      \node (U1) [utility, player1, right = of D1] {$U^1$};
      \node (U2) [utility, player2, right = of D2] {$U^2$};
      \node (D22) [decision, player2, above = of D1] {$D^2$};
      \node (U22) [utility, player2, right = of D22] {$U^2$};
      \node (D23) [decision, player2, right = of U22] {$D^2$};
      \node (U23) [utility, player2, right = of D23] {$U^2$};
      \node (U13) [above = of U23] {$U^1$};
      \node (J4) [below = of D23] {$J$};
      \node (D14) [decision, player1, right = of J4] {$T$};
      \node (D24) [decision, player2, right = of D14] {$A$};
      \node (U14) [utility, player1, below right = 1.4cm and 0.7cm of J4] {$U^T$};
      \node (U24) [utility, player2, right = of U14] {$U^A$};

      \edge {D14} {U14, U24};
      \edge {D24} {U14, U24};
      \edge [information] {J4} {D14};
      \edge {J4} {U14};
      \edge [information] {D14} {D24};

      \edge {D22} {U22};
      \edge {D23} {U23, U13};
      \edge {D1, D2} {U1};
      \edge[information] {D1} {D2};
      \edge {D1, D2} {U2};

  \end{influence-diagram}
  }
  \caption{}
  \label{fig:subgames:a}
\end{subfigure}
  \begin{subfigure}[b]{0.49\linewidth}
      \centering
        \resizebox{0.7\width}{!}{
      \begin{istgame}[scale=0.8]
      \xtdistance{15mm}{45mm}
      \istroot(0){$J$}
      \istb{\neg q}[al]
      \istb{q}[ar] 
      \endist
      \xtdistance{15mm}{20mm}
      \istroot(1)(0-1)<180, red!70>{$D^1_1$}
      \istb{\neg d_1}[al]
      \istb{d_1}[ar] 
      \endist
      \istroot(2)(0-2)<0, red!70>{$D^1_2$}
      \istb{\neg d_2}[al]
      \istb{d_2}[ar] 
      \endist
      \xtdistance{15mm}{10mm}
      \istroot(3)(1-1)<180, blue!70>{$D^2_1$}
      \istb{\neg a_1}[al]{(0,0)}
      \istb{a_1}[ar]{(2,1)} 
      \endist
      \istroot(4)(1-2)<180, blue!70>{$D^2_2$}
      \istb{\neg a_2}[al]{(0,0)}
      \istb{a_2}[ar]{(1,2)} 
      \endist
      \istroot(5)(2-1)<0,blue!70>{$D^2_3$}
      \istb{\neg a_1}[al]{(-1,0)}
      \istb{a_1}[ar]{(2,1)} 
      \endist
      \istroot(6)(2-2)<0,blue!70>{$D^2_4$}
      \istb{\neg a_2}[al]{(-1,0)}
      \istb{a_2}[ar]{(1,2)} 
      \endist
      \xtCInfoset(3)(5){\textcolor{blue!70}{$I^2_1$}}[above right]
      \xtCInfoset(4)(6){\textcolor{blue!70}{$I^2_2$}}[above left]
      \xtSubgameBox(0){(0)(1)(3-1)(3-2)(4-1)(4-2)(2)(5-1)(5-2)(6-1)(6-2)}[black,inner sep = 17pt, xshift=0pt, yshift=0pt]
      \end{istgame}}
      \caption{}
      \label{fig:subgames:b}
  \end{subfigure}
  \caption{(a) shows the four subdiagrams (three of which are `proper') of the MAID in Figure \ref{fig:taxi:a} and (b) shows the corresponding EFG in which none of the MAID's proper subgames can be recognised.
  }
  \label{fig:subgames}
\end{figure}

We now briefly introduce subgames (see \cite{causalgames}) for more details) because they simplify the presentation of some of our proofs in Appendix \ref{app:proofs}. Subgames in EFGs represent parts of the game that can be solved independently from the rest. In MAIDs, they fulfil the same purpose: they identify parts of the game that can be solved independently (and allow a subgame-perfect equilibrium refinement to be defined). Subgames in MAIDs are found by exploiting $s$-reachability to find the graphs underlying the subgames, called sub-diagrams. To then find the subgames for each subdiagram, the parameterisation of the remaining variables is updated to be consistent with the original game and graph structure. 

Importantly, because MAIDs explicitly represent conditional independencies between variables, we can often find more subgames in a MAID than in a corresponding EFG. This is the case for Example \ref{ex:taxi}'s MAID (shown in Figure \ref{fig:taxi:a}) with the four subdiagrams (three proper) in Figure \ref{fig:subgames:a}. Each subdiagram has a set of associated subgames, one for each instantiation of the variables outside of the subdiagram.
None of the proper MAID subgames can be recognised as subgames in the corresponding EFG (in Figure \ref{fig:subgames:b}).

\begin{definition}
  \label{def:subdiagram}
  Given a MAID $\model = (\graph, \bm{\theta})$, with $\graph = (N, \bm{V}, E)$, the subgraph $(\bm{V}', E')$ of $\graph$, along with the set of agents $N' \subseteq N$ possessing decision variables in that subgraph, is known as a \textbf{subdiagram} $\graph' = (N', \bm{V}', E')$ if:
  \begin{itemize}
    \item $\bm{V}'$ contains every variable $Z$ such that $\mecvar_Z$ is $s$-reachable from some $\Pi_{D}$ with $D \in \bm{V}'$,
    \item $\bm{V}'$ contains, for all $X,Y \in \bm{V}'$, every variable that lies on a directed path $X \pathto Y$ in $\graph$. 
  \end{itemize}
  A \textbf{subgame} of $\model$ is a new MAID $\model'=(\graph', \bm{\theta}')$ where $\graph'$ is a subdiagram of $\graph$ and
  $\bm{\theta}'$ is defined by $\Pr'(\bm{v}' ; \bm{\theta}') \coloneqq \Pr(\bm{v}' \mid \bm{z} ; \bm{\theta})$, where $\bm{z}$ is some instantiation of the variables $\bm{Z} = \bm{V}\setminus \bm{V'}$. A subgame is \textbf{feasible} if there exists a policy profile ${\bm{\pi}}$ where $\Pr^{\bm{\pi}}(\bm{z}) > 0$. 
\end{definition}

The first condition on $\bm{V}'$ ensures that for any decision variable $D$ in the subdiagram, any variable whose mechanism may impact the optimal decision rule for $D$ is also included in the graph. The second condition says that additional variables may also be included in the subdiagram as long as mediators are included too. This ensures that the CPDs for all the variables in the subgame remain consistent.

\section{Proofs}
\label{app:proofs}
\setcounter{lemma}{0}
\setcounter{proposition}{0}

\begin{proposition}
    Both forgetfulness and absent-mindedness can prevent the existence of an NE in behavioural policies.
\end{proposition}

\begin{proof}
  Example \ref{ex:forget} (Figures \ref{fig:noNE:a}-\ref{fig:noNE:c}) and Example \ref{ex:absentminded} (Figures \ref{fig:noNE:d}-\ref{fig:noNE:f}) are counterexamples for each case.
  
  \textbf{Proof for Example \ref{ex:forget} (forgetfulness):}
  The normal-form game showing the payoffs for each agent is shown in Figure \ref{fig:noNE:c}. First, observe that there are no NE in pure policies. Now, suppose that there does exist an NE in behavioural policies. If Alice always plays $a$ or always $\bar{a}$ -- i.e., $\pi^A(a)=1$ or $\pi^A(a)=0$ -- then Bob's best response is always $\bar{b}_1 \bar{b}_2$ or always $b_1 b_2$, respectively. However, this does not form an NE. So, Alice must select a stochastic decision rule $\pi_{A}$ and be indifferent (by the principle of indifference) between $a$ and $\bar{a}$. 
  
  Letting $\Pi_{B_1}$ and $\Pi_{B_2}$ be parameterised by $p, q \in [0,1]$ where $\pi_{B_1}(b_1)=p$ and $\pi_{B_2}(b_2)=q$, we obtain two constraints on $p$ and $q$. On the one hand, by virtue of Alice's indifference, Bob's behavioural policy ${\bm{\pi}}^B$ must result in $\pi^B(\neg b_1, \neg b_2) = \pi^B(b_1, b_2)$, and so:
  $(1-p)(1-q) = pq \implies p + q = 1$. On the other hand, Bob receives utility $-1$ if his policy ${\bm{\pi}}^B$ results in any outcome with $B_1=\neg b_1$ and $B_2=b_2$, or $B_1=b_1$ and $B_2=\neg b_2$, whatever the choice of ${\bm{\pi}}^A$.
  Therefore, we must have that $\pi^B(\neg b_1, b_2) + \pi^B(b_1, \neg b_2)  < \pi^B(b_1, b_2) + \pi^2(\neg b_1, \neg b_2)$ and thus, by substituting in the result that $p + q = 1$:
  $(1-p)q + p(1-q) < pq + (1-p)(1-q) \implies (2p - 1)^2 < 0.$. This contradiction implies that the MAID for Example \ref{ex:forget} has no NE in behavioural policies.
  
   To further understand this example, let us again write Bob's policy as a tuple $(p,q)$, and suppose $\pi_{A}(a) = 0.5$. Then, either pure policy $(1,1)$ and $(0,0)$ is a best response for Bob with $EU^B = 0$. But, consider the convex combination of these best responses $0.5\cdot (1,1)+0.5\cdot(0,0)=(0.5,0.5)$. Under this policy, each of the eight outcomes in the payoff matrix is equally likely and so Bob's expected payoff drops to $(-1-1-1+1+1-1-1-1)/8=-0.5$. Since a convex combination of best responses is no longer a best response, Bob's best response function is not convex-valued, and so nor is the grand best response function. The conditions of Kakutani's fixed point theorem are not satisfied, which explains why a Nash equilibrium need not exist.
 
  \textbf{Proof for Example \ref{ex:absentminded} (absent-mindedness):}
  First, observe from the normal-form game in Figure \ref{fig:noNE:f} that there is no NE in pure policies in this game. Next, suppose there exists a NE in behavioural policies and let $\Pi_{B}$ be parameterised by $p\in [0,1]$, where $\pi_{B}(b)=p$ for $p\in [0,1]$. Alice's payoff only depends on her policy $\pi^A$ when Bob plays $bb$ or $\bar{b}\bar{b}$, for which Alice has pure best responses. This implies that, at an NE, $p^2=(1-p)^2 \implies p=0.5$. Therefore, Alice's policy is irrelevant and $EU^B = -1$ ($EU^B = 0$)  if he does (doesn't) forfeit, which happens with probability $0.5$. Therefore, Bob's policy is dominated by his pure policies, with worst-case payoff $EU^B = -1$. This contradicts the assumption of an NE in behavioural policies.

  \emph{Explanation:} If $\pi_A(a) = 0.5$,  then $p=0$ and $p'=1$ are both best responses for Bob with $EU^B=0$. However, the convex combination $0.5p+0.5p'$ gives expected payoff to Bob $EU^B=0.25\cdot1+0.25\cdot(-1) + 0.5\cdot (-10)=-5$ and is therefore not a best response. Again this is due to the fact that under behavioural policies, in situations of imperfect recall, a convex combination of pure policies can introduce outcomes that could not occur under either pure policy. Under a mixed combination of pure policies, Alice will always follow one or the other, and so no new outcomes are introduced. However, under a behavioural combination, two independent absent-minded draws from the same distribution over actions can come out differently, introducing new potential outcomes\textemdash in this case forfeit. 
\end{proof}

\setcounter{lemma}{0}
\setcounter{proposition}{1}
\begin{proposition}
  Given a MAID $\model$ with any partial profile ${\bm{\pi}}^{-i}$ for agents $-i$, then if agent $i$ is not absent-minded, for any behavioural policy ${\bm{\pi}}^i$ there exists a pure policy ${\dot{\bm{\pi}}}^i$ which yields a payoff at least as high against ${\bm{\pi}}^{-i}$.
    On the other hand, if agent $i$ is absent-minded in $\model$ across a pair of decisions with descendants in $\bm{U}^i$, then there exists a parameterisation of $\model$ and a behavioural policy 
${\bm{\pi}}^i$ which yields a payoff strictly higher than any payoff achievable by a pure policy.
\end{proposition}

\begin{proof}    
  Let ${\bm{\pi}}^i$ be a behavioural policy and begin with any decision node $D \in \bm{D}^i$ with decision rule $\pi_D \in {\bm{\pi}}^i$. Now $\pi_D^i(d \mid \pa_D)$ is the probability of choosing $d \in \dom(D)$ at $D$ when $\Pa_D = \pa_D$ according to ${\bm{\pi}}^i$. Since agent $i$ is not absent-minded, the expected payoff for agent $i$ can be written $EU^i({\bm{\pi}}^i, {\bm{\pi}}^{-i}) = \sum_{d \in \dom(D)}\pi^i(d \mid \pa_D)\lambda_d + \nu$, where each coefficent $\lambda_d$ and $\nu$ are independent of $\pi_D^i(d \mid\pa_D)$. Consider the action $\hat{d} \in \dom(D)$ which achieves the highest  $\lambda_d$ (i.e., contributes most the expected utility) 
  Setting $\pi_D^i(\hat{d} \mid \pa_D)=1$ therefore yields a payoff at least as high. The first claim therefore follows by repeating this argument for every $D \in \bm{D}^i$.

  For the converse claim, agent $i$ is absent-minded, which means that at least two of agent $i$'s decision nodes must draw from an identical distribution. Without loss of generality, call these $D_l$ and $D_m$. Recall that for this to be the case, $dom(D_l)=dom(D_m)$ and $dom(\emph{\Pa}_{D_l})=dom(\emph{\Pa}_{D_m})$. Now consider an outcome of the game $\hat{\bm{v}} \in \dom(\bm{V})$ where $\pa_{D_l} = \pa_{D_m}$, but $d_l \neq d_m$. Since $D_l$ and $D_m$ have descendants in $\bm{U}^i$, Parameterise the MAID $\model$ such that $EU^i = 1$ if and only if $\bm{V} = \hat{\bm{v}}$. For all other game outcomes $\bm{v} \neq \hat{\bm{v}}$, let $EU^i = 0$. The claim follows since the outcome $\hat{\bm{v}}$ cannot be instantiated by any pure policy for agent $i$, but can be instantiated by any behavioural policy for agent $i$ that has a (shared) decision rule for  $D_l$ and $D_m$ that assigns a positive probability to both actions $d_l$ and $d_m$. 
\end{proof}

\begin{proposition}
  A MAID with sufficient information always has an NE in pure policies, a MAID with sufficient recall always has an NE in behavioural policies, and every MAID has an NE in mixed policies.
\end{proposition}

\begin{proof}
    The mixed policies case follows from Nash's theorem since all the finite number of random variables in a MAID have finite domains \cite{nash1950equilibrium}. Hammond et al. proved the case with sufficient recall~\cite{causalgames}.
    
    We now consider the sufficient information case where we show that a NE in pure policies must exist. Begin with an arbitrary policy profile across all decision nodes in the original MAID, $\model$.  Decision rules associated with each $D \in \bm{D}$ can be optimised by iterating backwards through a subdiagram ordering $\graph_1 \prec \dots \prec \graph_m$ of $\model$'s subdiagrams such that $\graph_j \prec \graph_k$ implies that $\graph_j$ is \emph{not} a subdiagram of $\graph_k$. When $\model$ is a sufficient information game, this means that $\graph_m$ contains just one decision node for some agent $i \in N$, and, for each subdiagram $\graph_j$ where $1 \leq j < m$, $\graph_{j-1}$ contains \emph{at most} one additional decision variable. Several subdiagrams can have the same set of decisions, $\bm{D}_k$, so we choose a single subdiagram $\graph_k$ (one with the fewest nodes $\bm{V}'$) for each $\bm{D}_k$ and discard the others. Each subdiagram in this ordering has an associated subgame for each setting of the nodes which have a child in $\bm{V}'$. 

    When considering each subgame $\model_{m-j}$ for $\graph_{m-j}$, the decision rules for all decision nodes in proper subgames of $\model_{m-j}$ will have already been optimised and fixed in previous iterations, so these are now chance nodes in $\model_{m-j}$. In addition, the decision node $D_{m-j}$ in $\model_{m-j}$ does not strategically rely on any of the decision nodes outside of $\model_{m-j}$. Therefore, this step is localised to computing only the optimal decision rule for $D_{m-j}$. Since this is a single-agent single-decision optimisation, we know that there must exist a pure decision rule best response. In the case of a tie, pick one arbitrarily. After repeating this optimisation process for all subgames in the MAID, we know that every decision node must have a pure decision rule, so we have found a NE in pure policies, as required.
\end{proof}

\begin{proposition}
A MAID-CE in bounded treewidth MAIDs with sufficient recall can be found in poly-time.
\end{proposition}

\begin{proof}[Proof sketch]
    We follow Huang and von Stengel's method for this result \cite{huang2008computing}. Our result comes from the observation that if there is sufficient recall in a MAID, then: (i) the set of decision contexts of every decision node in the MAID is in bijection with the set of all information sets in a corresponding EFG; and (ii) sufficient recall is sufficient for the ordering of decision contexts analogous to Huang and von Stengel's ordering of information sets.
\end{proof}

\begin{lemma}
  \label{lemma}
  If \textsc{Is-Best-Response} can be solved in poly-time, then agent $i$'s expected utility under a best response to a partial policy profile ${\bm{\pi}}^{-i}$ in a MAID can be found in poly-time.
\end{lemma}
\begin{proof}
  This follows immediately from using binary search over agent $i$'s policies and uses the fact that we are restricting parameters in the MAID to be rational numbers.
\end{proof}

\setcounter{proposition}{6}

\begin{proposition}
If the in-degrees of $\bm{D}^i$ are bounded and \textsc{Is-Best-Response} can be solved in poly-time, then a best response policy for agent $i$ to a partial policy profile ${\bm{\pi}}^{-i}$ can be found in poly-time.
\end{proposition}

\begin{proof}
  Begin by constructing the MAID $\model({\bm{\pi}}^{-i})$
 by replacing decision nodes $\bm{D} \setminus \bm{D}^{-i}$ as chance nodes with CPDs given by ${\bm{\pi}}^{-i}$. Next, use Lemma \ref{lemma} to compute agent $i$'s expected utility under a best response policy in $\model({\bm{\pi}}^{-i})$ and use this value as $q$.
  Take each of agent $i$'s decision variables $D \in \bm{D}^i$ and build a new MAID $\model({\bm{\pi}}^{-i}, \pi_D)$ for every possible decision rule of $D$ (i.e., replace $D$ as a chance node with CPD $\pi_D$). The fact that the in-degrees of agent $i$'s decision nodes are bounded, bounds the number of these MAIDs. For each induced MAID, we can then use a poly-time algorithm for \textsc{Is-Best-Response} to determine any decision rule $\pi_D$ that makes up the best response policy for agent~$i$. 
\end{proof}

\setcounter{proposition}{9}
  \begin{proposition}
    In a MAID with sufficient information, if the in-degrees of $\bm{D}$ are bounded and \textsc{Is-Best-Response} can be solved in poly-time, then a pure NE can be found in poly-time.
  \end{proposition}

\begin{proof}
  First, note that we can check whether a MAID is a sufficient information game in poly-time using $s$-reachability, a graphical criterion based on d-separation \cite{shachter2013bayes}. We can then follow the constructive procedure given for the proof of Proposition \ref{prop:suffexistence}. Given Proposition \ref{prop:optimalbestresponse}, each optimisation step must take poly-time and since the in-degrees of all decision nodes are bounded by a constant, the number of subgames is also bounded by a constant. Therefore, the entire procedure takes poly-time.
\end{proof}

\end{document}